\documentclass{amsart}
\usepackage{amssymb}
\usepackage{amsthm}
\usepackage{amstext}
\usepackage{amsbsy}
\usepackage{amscd}
\usepackage{enumerate}
\usepackage{subfigure}
\usepackage[noadjust]{cite}
\usepackage[linktocpage]{hyperref}
\usepackage{graphicx}
\usepackage{mathtools}
\usepackage{pict2e}
\usepackage[usenames,dvipsnames]{color}
\usepackage[toc,page]{appendix}
\usepackage{tikz}
\usepackage{verbatim}
\usetikzlibrary{matrix}

\newtheorem{thm}{Theorem}[section]
\newtheorem{cor}[thm]{Corollary}
\newtheorem{lem}[thm]{Lemma}
\newtheorem{prop}[thm]{Proposition}

\theoremstyle{remark}
\newtheorem{rem}[thm]{Remark}
\theoremstyle{definition}
\newtheorem{defn}[thm]{Definition}

\newcommand{\norm}[1]{\left\Vert#1\right\Vert}
\newcommand{\abs}[1]{\left\vert#1\right\vert}
\newcommand{\set}[1]{\left\{#1\right\}}
\newcommand{\R}{\mathbb{R}}
\newcommand{\C}{\mathbb{C}}
\newcommand{\N}{\mathbb{N}}
\newcommand{\Z}{\mathbb{Z}}
\newcommand{\Tr}{\mathrm{Tr}}

\renewcommand{\dim}{\mathbf{dim}}
\newcommand{\hdim}{\dim_\mathrm{H}}
\newcommand{\bdim}{\dim_\mathrm{B}}
\newcommand{\loc}{\mathrm{loc}}
\newcommand{\lhdim}{\hdim^\loc}
\newcommand{\lbdim}{\bdim^\loc}

\newcommand{\eqdef}{\overset{\mathrm{def}}=}

\newcommand{\Int}{\mathrm{Int}}
\newcommand{\SL}{\mathrm{SL}}

\begin{document}

\title[Tridiagonal substitution Hamiltonians]{Tridiagonal substitution Hamiltonians}

\author[M. Mei]{May Mei}
\email{meim@denison.edu}
\address{Mathematics \& Computer Science, Denison University, Granville, OH 43023-0810}

\author[W. Yessen]{William Yessen}
\email{yessen@rice.edu}
\address{Mathematics, Rice University, 1600 Main St. MS-136, Houston, TX 77005}

\thanks{Part of the work presented herein was supported by DMS-0901627 (PI: A. Gorodetski).
\newline
\indent W.Y. was supported by the NSF grant DMS-1304287.\\\indent M.M. was supported by the Michele T. Myers PD Account through Denison University.
\newline
\newline
\indent The body of this paper contains text, mathematical formulas, statements and proofs of mathematical results, figures, plots and diagrams that significantly overlap with the authors' previous work in \cite{Yessen2011, Yessen2011a, Yessen2012a, Damanik2013a, Damanik2013b, Mei2013}; this previous work was supported by DMS-0901627 (PI: A. Gorodetski), and the authors gratefully acknowledge this support.}
\subjclass[2010]{47B36, 82B44.}

\date{\today}

\begin{abstract}

We consider a family of discrete Jacobi operators on the one-dimensional integer lattice with Laplacian and potential terms modulated by a primitive invertible two-letter substitution. We investigate the spectrum and the spectral type, the fractal structure and fractal dimensions of the spectrum, exact dimensionality of the integrated density of states, and the gap structure. We present a review of previous results, some applications, and open problems. Our investigation is based largely on the dynamics of trace maps. This work is an extension of similar results on Schr\"odinger operators, although some of the results that we obtain differ qualitatively and quantitatively from those for the Schr\"odinger operators. The nontrivialities of this extension lie in the dynamics of the associated trace map as one attempts to extend the trace map formalism from the Schr\"odinger cocycle to the Jacobi one. In fact, the Jacobi operators considered here are, in a sense, a test item, as many other models can be attacked via the same techniques, and we present an extensive discussion on this.
\end{abstract}

\maketitle

\section{Introduction}\label{sec:p-1}

\subsection{The setting}\label{sec:intro}

In this paper we are concerned with spectral properties of a class of bounded self-adjoint operators acting on the space of complex-valued square-summable sequences over the one-dimensional integer lattice, $\ell^2(\Z, \C)$. More precisely, we are interested in \textit{Jacobi operators} that are defined as follows. Given bounded real-valued sequences $a = \set{a_n}_{n\in\Z}$ and $b = \set{b_n}_{n\in\Z}$, with $a_n\neq 0$ for all $n$, the Jacobi operator $H_{a, b}: \ell^2(\Z, \C)\rightarrow \ell^2(\Z, \C)$ is defined by
\begin{align}\label{eq:gen-jacobi}
  (H_{a, b}\phi)_n
    =
  a_{n+1}\phi_{n+1} + a_n\phi_{n-1} + b_n\phi_n.
\end{align}
In what follows, we often drop $a, b$ from notation, dependence of $H$ on $a$ and $b$ being implicitly understood. We also drop $\C$ from the notation and write simply $\ell^2(\Z)$.

Jacobi operators play a central role in several branches of mathematics. For example, these operators are canonical representatives of self-adjoint operators, and are also intimately related to orthogonal polynomials on the real line (for more details, see, for example, the introduction of \cite{Simon2007} and references therein).

Jacobi operators also arise quite naturally in mathematical physics, particularly in quantum mechanics. One-dimensional Jacobi operators are energy Hamiltonians of quantum particles in one-dimensional media. In this case the \textit{discrete Laplacian} $a_{n+1}\phi_{n+1} + a_n\phi_{n-1}$ represents the kinetic energy, for example due to external forces such as a magnetic field, and $b_n\phi_n$ represents the potential energy. Indeed in case $a_n = 1$ for all $n$, the operator $H$ is known as the discrete variant of the one-dimensional \textit{Schr\"odinger operator}. For a textbook exposition on these topics, we refer the reader to \cite{Teschl1999}, Chapters 7-10 of \cite{Teschl2009}, Chapter 3 of \cite{Takhtajan2008}, and Chapters 4-5 of \cite{Hall2013}. Notice that in our case, the coefficients $\set{a_n}$ are not necessarily identically equal to 1. This more general case of a \textit{weighted} Laplacian appears in a number of applications. For example, the Harper's model \cite{Harper1955, Luttinger1951}. Also, using the Jordan-Wigner transformation together with the Lieb-Schultz-Mattis method \cite{Lieb1961}, it is possible to reduce some quantum many-body problems to the study of operators \eqref{eq:gen-jacobi}; see, for example, the recent studies \cite{Yessen2011,Yessen2012a} and references therein (we single out these papers as they focus on aperiodicity in the models as we do in the present paper, and they contain extensive introductions with relevant references; the same techniques, however, are applicable in general \cite{Lieb1961}). In particular, \cite{Yessen2012a} focuses on the study of the energy spectra of Ising models with Fibonacci disorder (the Fibonacci disorder is explained in detail in later sections) via the operator \eqref{eq:gen-jacobi}, relying heavily on the results of \cite{Yessen2011a}.

Spectral analysis of operators \eqref{eq:gen-jacobi} is also important in connection with quantum dynamics of some quite natural systems. For example, see \cite{Hamza2012} for a recent study of Heisenberg dynamics of some many-body quantum systems with random disorder, where localization is demonstrated by reducing the problem to a (block) Jacobi operator and using results from the Anderson localization theory.

In general, time evolution of a quantum particle whose time-dependent state is represented by $\Psi: \R\rightarrow\ell^2(\Z)$ is modeled by the Schr\"odinger's equation,
\begin{align}\label{eq:schro}
  i\frac{\partial}{\partial t}\Psi = H\Psi,
  \hspace{2mm} \text{with solution} \hspace{2mm}
  \Psi(t) = e^{-itH}\Psi(0).
\end{align}
It turns out that for the description of quantum dynamics of the system given in \eqref{eq:schro} detailed information about the spectrum of $H$ as well as its spectral type is required. Indeed, this is not at all surprising, since by the spectral theorem, we have
\begin{align*}
 \langle \delta_i, g(H)\delta_i\rangle = \int g(E)d\mu_i(E)
\end{align*}
for any bounded and measurable $g$, where $\mu_i$ is the spectral measure associated to $\delta_i$, where $\set{\delta_i}_{i\in\Z}$ is the canonical basis of $\ell^2(\Z)$. Further details on the general theory can be found in, for example, \cite{Oliveira2009}. A detailed study of different modes of quantum transport in one-dimensional systems is presented in \cite{Damanik2001}.

Spectral properties of certain classes of the operator \eqref{eq:gen-jacobi} (the class being determined by properties of the sequences $a_n$ and $b_n$) are fairly well understood. For example, spectral properties and quantum dynamics of $H_{a,b}$ with periodic sequences $a_n$ and $b_n$ (the \textit{periodic Jacobi operators}) are completely understood via Floquet theory (for a somewhat broad exposition, see, for example, \cite{Toda1981} or, for a more modern overview, \cite{Teschl1999}). On the other hand, \textit{random Jacobi operators} have been under heavy investigation for the past few decades and while many questions are still open, a great deal is understood (the literature in this field is voluminous; we point to the textbook expositions in \cite{Carmona1990, Cycon1987, Pastur1992} for a general overview). In this paper we are interested yet in a third class, which is somewhat intermediate between the random and the periodic cases (to be rigorously described later). This so-called \textit{aperiodic} case has attracted much attention since early 1980's, greatly motivated by mathematical and physical investigation of quasicrystals (a material that was discovered by D. Shechtman et. al. in the early 1980's and for which Shechtman received the Nobel prize in chemistry in 2012 \cite{Shechtman1984}).

In this paper we study the spectrum and spectral type of discrete \textit{substitution} Jacobi operators where the Laplacian terms and the potential terms follow a \textit{substitution sequence} that is generated by a \textit{primitive, invertible substitution} on two letters $\set{0, 1}$. As previously mentioned, such sequences are neither periodic nor random, but somewhere in between. To make this notion precise, let $P:\N\times\set{0, 1}^\Z\rightarrow \N$ denote the complexity function, that is, $P(k,\omega)$ is the number of distinct strings $w_1w_2\cdots w_k$, or \textit{words}, of length $k$ that appear in the infinite sequence $\omega = \cdots \omega_{-1}\omega_0\omega_1\cdots$. Then $\omega$ is \textit{eventually periodic} (i.e. $\set{\omega_k}$, with $k \in (-\infty, -N]\cup[N, \infty)$ for some $N\in\N$, is periodic) if and only if $P(k, \omega) = k$ for some $k$. Note that if $P(k,\omega)=k$ for some $k$ then $P(k',\omega)=k$ for all $k'>k$. In contrast, if $\omega$ is a random sequence, $P(k, \omega) = 2^k$. The sequences that that we consider have complexity $P(k, \omega) = k+1$. Thus such sequences are nonperiodic sequences of minimal complexity. It turns out that there are a few equivalent constructive ways of producing such sequences (for example, see \cite{Fogg2002} for comprehensive exposition; see also \cite{MorseHedlund1940}); we will be using two constructions: substitutions and sampling of irrational circle rotations, both of which we now describe.

Let us denote by $\mathcal{A}$ the set $\set{0, 1}$, to serve as our \textit{alphabet}, and by $\mathcal{A}^*$ the set of all finite words over $\mathcal{A}$ (i.e. finite sequences consisting of letters $0$ and $1$). A map $s: \mathcal{A}\rightarrow\mathcal{A}^*$ is called a \textit{substitution}. The substitution $s$ is extended to a morphism on the free monoid generated by $0$ and $1$, which we can also denote by $\mathcal{A}^*$, by concatenation. For any finite \textit{word} $\omega_1\cdots\omega_n\in\mathcal{A}^*$, $s(\omega_1\cdots\omega_n) = s(\omega_1)\cdots s(\omega_n)$. For $k\in\N$, by $s^k$ we mean the $k$-fold composition of $s$ with itself. In this paper we will be concerned only with those substitutions which satisfy the following two conditions.
\begin{itemize}
\setlength{\itemsep}{2pt}

\item \textit{Primitivity}: there exists $k\in\N$ such that $s^k(0)$ and $s^k(1)$ both contain $0$ and $1$.

\item \textit{Invertibility}: the substitution $s$ extends to an isomorphism on the free group $\langle 0,1\rangle$.

\end{itemize}
Henceforth all substitutions are implicitly assumed to posses those two properties.
\begin{rem}
We should remark that some primitive noninvertible substitutions have also been considered, for example in \cite{Bellissard1990,Bellissard1989,Bellissard1991}. Our techniques, however, rely on invertibility in a fundamental way and cannot be generalized to noninvertible cases. In fact, much less is currently known in noninvertible cases than in the invertible. Also, substitutions on three or more letters can be considered, see \cite{WenZhang1999} and \cite{TanWenZhang2003}. In this situation some of the geometric and dynamical systems techniques that are exploited herein can be generalized, but the resulting dynamical systems are of higher dimension; as far as we know, these resulting higher dimensional polynomial dynamical systems have not yet been given a detailed treatment (at least in the context of spectral theory).

Also, a more general case of quasi-Sturmian sequences has been considered; see \cite{DamanikLenz2003b} for further detail.
\end{rem}
It is well known that primitivity guarantees the following. If $s(0)$ begins with $0$ or if $s(1)$ begins with $1$, then there exists a fixed point $\omega\in\set{0,1}^\N$ of $s$; that is, $s(\omega) = \omega$ ($s$ is extended to an action on $\set{0,1}^\mathbb{M}$ with $\mathbb{M} = \N$ or, later in this paper, $\mathbb{M} = \Z$, by concatenation). Then for any primitive substitution $s$ (i.e. not necessarily one for which $s(0)$ begins with $0$ or $s(1)$ begins with $1$), there exists a fixed point $\omega\in\set{0,1}$ for $s^l$, with some $l\leq 2$ (indeed, if $s(0)$ begins with $1$ and $s(1)$ begins with $0$, then $s^2(1)$ begins with $1$ and $s^2(0)$ begins with $0$). In fact, a fixed point $\omega$ can be constructed as follows. Assuming that $\star\in\set{0,1}$ is such that $s^l(\star)$ begins with $\star$, consider
\begin{align*}
\star\mapsto s^l(\star) \mapsto s^{2l}(\star) \mapsto \cdots
\end{align*}
Then the word $s^{ml}(\star)$ begins with the word $s^{(m-1)l}(\star)$. Hence inductively a sequence $\omega$ in $\set{0,1}^\N$ is obtained, with $s^l(\omega) = \omega$. This sequence is not periodic, and it is known that its complexity is minimal (i.e. $P(k, \omega) = k+1$). There is a way to extend this sequence to the left that is natural in the sense that its left half exhibits the same sense of ``order" (more rigorously: the same recurrence of finite words) as its right half. One way to do so is via sampling of certain irrational circle rotations, which we describe next.

Let $\mathbb{T}$ denote the unit circle. It is known from \cite{CrispMoranPollingtonShiue1993} that for a given primitive and invertible substitution $s$, there exists $\alpha\in (0, 1)$ irrational and $\beta\in \mathbb{T}$ such that the fixed point $\omega = \omega_1\omega_2\omega_3\cdots$ of $s^l$ is given by
\begin{align}\label{eq:sampling}
\omega_n
  =
 \chi_{[1 - \alpha, 1)}(n\alpha + \beta \mod 1),\hspace{2mm}n\in\N,
\end{align}
where $\chi$ is the $\set{0, 1}$-valued characteristic function; in this case $\omega$ is called a rotation sequence. To be precise, we must also include the statement above with alternative characteristic function $\chi_{(1-\alpha, 1]}$. In fact, every two-letter non-periodic sequence with minimal complexity (or, as they are often called, \textit{Sturmian sequence}) can be realized as a rotation sequence. Thus $\omega$ is obtained by sampling the forward orbit of an irrational circle rotation. By sampling the backward orbit, we obtain the left part of $\omega$; thus sampling of the full orbit gives an extension of $\omega$ to the left.

While every primitive and invertible substitution sequence (rather, a fixed point of such a substitution) is realizable as a rotation sequence with $\alpha$ as above, not every such rotation sequence can be realized as a fixed point of a primitive invertible substitution. This isn't difficult to see, since the set of all primitive invertible two-letter substitutions is countable, while for uncountably many irrational $\alpha$, one obtains distinct rotation sequences. In fact, more is known: those rotation sequences that can be realized as a fixed point of a primitive invertible two-letter substitution form a subset of the sequences for which the coefficients of the continued fraction expansion of the corresponding rotation number $\alpha$ form an eventually periodic sequence.  For a given substitution, such $\alpha$ can be explicitly found through the methods in \cite{CrispMoranPollingtonShiue1993}. However, not every rotation sequence whose rotation number $\alpha$ has an eventually periodic continued fraction expansion is fixed by a primitive and invertible substitution, see \cite{Brown1991} for an explicit example. All our results are stated and proven for primitive and invertible substitutions. Henceforth, the term \textit{substitution sequence} will refer to a sequence of this type.

Now to each substitution sequence $v$ we associate a Jacobi operator (in fact, a family thereof) which will be the focus of the rest of the paper. Given a substitution sequence $v\in\set{0,1}^\Z$, let $\Omega_v$ denote the hull of $v$ defined by
\begin{align}\label{eq:hull}
\Omega_v
  \eqdef
 \set{\nu\in \set{0, 1}^\Z: \nu = \lim_{i\rightarrow\infty}T^{n_i}(v),\hspace{2mm}n_i\uparrow\infty},
\end{align}
where $T:\Omega_v\rightarrow\Omega_v$ is the left shift (i.e. $(T\omega)_n = \omega_{n+1}$). We will write $\Omega$ instead of $\Omega_v$ when the meaning is clear from the context. It isn't difficult to see that $\Omega$ is compact and $T$-invariant, where the topology on $\Omega$ is the induced product topology of $\set{0, 1}^\Z$ (the set $\set{0, 1}$ is taken with the discrete topology); moreover, $v\in\Omega$. It is also true that $T$ is minimal, that is, for every $\omega\in\Omega$, the forward orbit of $\omega$, $\set{T^n(\omega)}_{n\in\N}$, is dense in $\Omega$. Furthermore, $T$ is uniquely ergodic, that is, there exists a (unique) Borel probability measure on $\Omega$ with respect to which $T$ is ergodic (minimality follows from Gottschalk's theorem \cite{Gottschalk1963} and ergodicity follows from Oxtoby's ergodic theorem \cite{Oxtoby1952}).

Now let $p: \set{0, 1}\rightarrow \R$ such that $p(0)$ and $p(1)$ are nonzero, and $q:\set{0,1}\rightarrow\R$. For any $\omega$ in $\Omega_v$, define the operator $H_{\omega, (p, q)}: \ell^2(\Z)\rightarrow\ell^2(\Z)$ as follows.
\begin{align}\label{eq:qp-Jacobi}
(H_{\omega, (p, q)}\phi)_n
  =
 p(\omega_{n + 1})\phi_{n+1} + p(\omega_n)\phi_{n-1} + q(\omega_n)\phi_n.
\end{align}
Clearly $H_{\omega, (p, q)}$ is bounded and self-adjoint. Our aim in this paper is to describe the spectral properties of such operators.

\subsection{Organization of the paper}

The paper is organized as follows. In Section \ref{sec:survey} we give a detailed overview of the theory, including many recent contributions. Specifically, in Section \ref{sec:known} we discuss the known results on the Fibonacci Schr\"odinger operator. Section \ref{sec:generalization-mei} contains an overview of the extension of the results from the Fibonacci Schr\"odinger operator to the one with potential modulated by a general primitive invertible two-letter substitution. The case of the Fibonacci Jacobi Hamiltonian is discussed in Section \ref{sec:generalization-yessen}. The main new results of this paper are presented in Section \ref{sec:main}, and the proofs of the results are given in Section \ref{sec:proofs}.
\subsubsection*{A note on notation}

Many of the results and the proofs presented in this paper are quite technical and complicated notation is inevitable. Moreover, there are a few papers that we cite, the results (and proofs thereof) of which we rely on heavily; for this reason we keep our notation synchronized with that used in those papers. In other places, we have tried to follow the well-established conventions (for example, using the letter $\rho$ for the resolvent). For all of these reasons, there are a few symbols that appear in different context throughout the paper, and denote different things. For example, there is a place where $\rho$ is used to denote a representation of a free group into $\SL(2,\C)$, in another place it denotes the resolvent set of an operator, and yet in another place it denotes a certain curve in $\R^3$. In this and other similar cases, the parts where the same symbol appears but denotes different things are contextually independent and the notation is clearly (re)defined.

\section{A survey of previous results}\label{sec:survey}

\subsection{The Fibonacci Schr\"odinger case}\label{sec:known}

As has already been mentioned earlier in the introduction, investigation of substitution operators began in the early 1980's, motivated significantly by an interest in physical properties of the newly discovered quasicrystals. The pioneering mathematical work on the subject first appeared in the papers of Kohmoto et. al. and Ostlund et. al. \cite{Kohmoto1983, Ostlund1983} (the research by those two groups was conducted independently). In those papers, the two groups considered a tight-binding model given by the  Schr\"odinger Hamiltonian $H: \ell^2(\Z)\rightarrow \ell^2(\Z)$,
\begin{align}\label{eq:schrodinger}
(H\phi)_n
  =
 \phi_{n-1} + \phi_{n+1} + V_n\phi_n\hspace{2mm}\text{ where }\hspace{2mm} V_n = Vu_n,
\end{align}
with $V > 0$ and the sequence $\set{u_n}$ the fixed point of the (by now widely studied and well known) Fibonacci substitution, given by
\begin{align}\label{eq:fib-sub}
s: 0\mapsto 01,\hspace{2mm}\text{and}\hspace{2mm}s:1\mapsto 0.
\end{align}
The first strong results regarding the spectrum as well as the spectral type of the operator from \eqref{eq:schrodinger} began appearing in the late 1980's in the works of M. Casdagli and A. S\"ut\H{o} in \cite{Casdagli1986} and \cite{Suto1987}, using the method of trace maps that had been established in \cite{Kohmoto1983, Ostlund1983}. Let us describe these methods and recall some results about the Schr\"odinger operators (especially since some of these results are extended in this paper to Jacobi operators).

Let us reconsider the operator from \eqref{eq:schrodinger} in the following eigenvalue equation
\begin{align}\label{eq:eigen-schro}
 H\theta = E\theta.
\end{align}
Notice that for $E\in\R$, a vector $\phi \in \R^\Z$ solves the equation if and only if
\begin{align}\label{eq:transfer-schro}
 \begin{pmatrix}
  \theta_{n+1}\\
  \theta_n
 \end{pmatrix}
 =
 \begin{pmatrix}
  E - V_n		&		-1\\
  1			&		0
 \end{pmatrix}
 \begin{pmatrix}
  \theta_n\\
  \theta_{n-1}
 \end{pmatrix}.
\end{align}
Let us call the $2\times 2$ matrix from equation \eqref{eq:transfer-schro}, $M_n(E)$. The matrix $M_n(E)$ is called the \textit{transfer matrix at site $n$}. If we define the \textit{transfer matrix over $n$ sites}, $\hat{M}_n(E)$, as $\hat{M}_n(E)\eqdef M_n(E)M_{n-1}(E)\cdots M_1(E)$, then the following equation is clear from \eqref{eq:transfer-schro}:
\begin{align}
 \begin{pmatrix}
  \theta_{n+1}\\
  \theta_n
 \end{pmatrix}
 =
 \hat{M}_n(E)
 \begin{pmatrix}
  \theta_1\\
  \theta_0
 \end{pmatrix}.
\end{align}
It turns out that all of the information about the spectrum as a set is encoded in the traces of the transfer matrices $\set{\hat{M}_n}_{n\in\N}$. In the periodic case this is particularly easy to state: suppose that we replace the sequence $V_n$ with a sequence ${V}_n^{(k)}$ which is defined by periodically repeating the first word of length $k$ of $\set{V_n}$; that is, $V_1V_2\cdots V_k$. If $\hat{M}_n^{(k)}(E)$ denotes the corresponding transfer matrix, then $E$ is in the spectrum of the associated operator if and only if $\abs{\Tr \hat{M}_k^{(k)}(E)}\leq 2$ (this is a consequence of Floquet theory \cite{Toda1981}). Let us denote the corresponding periodic operator by $H^{(k)}$, and its spectrum by $\sigma_k$ (which consists of finitely many compact intervals, since the trace of $\hat{M}_k^{(k)}$ is, of course, a polynomial in $E$). It isn't difficult to see that $H$ is the strong limit of the sequence $\set{H^{(k)}}_{k\in\N}$. It is natural to inquire whether $\sigma_k$ approximate the spectrum of $H$, which we denote by $\sigma$, in the limit. Of course, strong convergence alone does not guarantee convergence of spectra. We do, however, have semicontinuity:
\begin{align}\label{eq:semicont-spectra}
 \sigma \subset \bigcap_{l\in\N}\overline{\bigcup_{k\geq l}\sigma_k}.
\end{align}
The semicontinuity from \eqref{eq:semicont-spectra} will play an important role in the proofs of our main results and we shall return to it later. Let us now define the set
\begin{align}\label{eq:binfty-def}
 B_\infty
  =
 \set{E\in\R: \abs{\hat{M}_n(E)}_{n\in\N}\hspace{2mm}\text{is bounded}}.
\end{align}
Applying the so-called \textit{Gordon argument} (which we shall discuss in detail in Section \ref{sec:proofs}), A. S\"ut\H{o} in \cite{Suto1987} was able to prove
\begin{thm}[A. S\"ut\H{o} 1987]\label{thm:suto}
 For the Schr\"odinger operator with the potential modulated by the Fibonacci substitution sequence, the spectrum, $\sigma$, coincides with the set $B_\infty$.
\end{thm}
\begin{rem}
 This result was later substantially generalized by D. Damanik \cite{Damanik2000}; see Theorem \ref{thm:damanik} in Section \ref{sec:generalization-mei} below.
\end{rem}

In fact, A. S\"ut\H{o} proved more, by showing that $B_\infty$ is a Cantor set for all $\abs{V} \geq 4$ (recall: $V_n = Vu_n$ where $u$ is the Fibonacci substitution sequence) and that for all $V$, $\sigma$ does not contain any point spectrum. M. Casdagli, however, had investigated the set $B_\infty$ a year earlier, in 1986, for $V \geq 16$ in his paper \cite{Casdagli1986}, but was unable to prove the equality $\sigma = B_\infty$. Using \eqref{eq:semicont-spectra}, on the other hand, one can get $\sigma\subset B_\infty$ (we shall show this later), and Casdagli was aware of this. The insight of Casdagli in the aforementioned paper, however, had opened the door to an all out campaign to solve the spectral problem for the substitution (in particular, Fibonacci) Schr\"odinger operators: the relationship between $B_\infty$ and a certain dynamical system had been established in the earlier works \cite{Kohmoto1983, Ostlund1983}, but Casdagli was the first to investigate the dynamics and derive very strong results
about $B_\infty$ (again, for large $V$) regarding its fractal structure. Let us review this technique next.

Notice that due to the recursive nature of the Fibonacci substitution sequence, the matrices $\hat{M}_n(E)$ also follow the same recursion (from right to left, since we take the product of the matrices from right to left):
\begin{align*}
 \hat{M}_{F_{k+1}}(E)
  =
 \hat{M}_{F_{k - 1}}(E)\hat{M}_{F_k}(E),
\end{align*}
where $F_k$ is the $k$th Fibonacci number. On the other hand, since the matrices $\hat{M}_n(E)$ are members of the $\SL(2,\C)$ group, the trace of $\hat{M}_{F_k}(E)$ can be computed recursively as a function of the traces of the previous three matrices (indexed by $F_{k-1}, F_{k-2}$ and $F_{k-3}$). The precise form of this recursive relation is \cite{Kohmoto1983, Ostlund1983}
\begin{align}\label{eq:fib-tmap}
 (x_{k+2}, x_{k+1}, x_k)
  =
 f(x_{k+1}, x_{k}, x_{k-1}),\hspace{2mm}\text{ with }\hspace{2mm} f(x,y,z) = (2xy - z, x, y),
\end{align}
and $x_k = x_k(E) \eqdef \frac{1}{2}\abs{\Tr\hat{M}_{F_k}(E)}$. The map $f$ is called the \textit{Fibonacci trace map}. In general, there is a trace map associated to every primitive and invertible substitution (we shall return to this in later sections). It is known that an unbounded forward orbit under $f$ does not contain a bounded suborbit (a subsequence of the orbit) \cite{Casdagli1986}. Thus we have
\begin{align*}
 B_\infty
    &=
	\set{E\in\R: \frac{1}{2}\abs{\hat{M}_{F_k}(E)}_{k\in\N}\hspace{2mm}\text{is bounded}}\\
    &=
	\set{E\in\R: \set{f^n(\gamma(E))}_{n\in\N}\hspace{2mm}\text{is bounded}},
\end{align*}
where $f^n$ denotes $n$-fold composition of $f$ with itself. Here $\gamma(E)$ denotes the \textit{line of initial conditions}, parameterized by $E\in\R$:
\begin{align}\label{eq:gamma-schro}
 \gamma(E)
  =
 \left(\frac{E - V}{2}, \frac{E}{2}, 1\right),
\end{align}
the coordinates of which are obtained as traces of the initial three matrices (actually, traces of $\hat{M}_1, \hat{M}_2$ and $\hat{M}_3$ are rather complicated expressions, so to obtain $\gamma(E)$ in its simpler form, one takes $f^{-1}(x_1, x_2, x_3)$, where $f^{-1}$ is the inverse of $f$ given by $f^{-1}(x,y,z) = (z, y, 2yz - x)$).

Henceforth we shall refer to the sequence $\set{f^n(x)}$ as the \textit{forward} or \textit{positive} semi-orbit of $x$ under $f$, and denote it by $\mathcal{O}_f^+(x)$. The \textit{backward} or the \textit{negative} semi-orbit is defined similarly by replacing $n$ with $-n$, and is denoted by $\mathcal{O}_f^-(x)$. The \textit{full} orbit is defined as $\mathcal{O}_f(x)\eqdef \mathcal{O}_f^+(x)\cup \mathcal{O}_f^-(x)$.

It turns out that $f$ (and every trace map associated to a primitive and invertible substitution on two letters) preserves the so-called \textit{Fricke-Vogt invariant}, given by
\begin{align}\label{eq:fv-invariant}
 I(x,y,z) \eqdef x^2 + y^2 + z^2 - 2xyz - 1.
\end{align}
More precisely, we have $I\circ f(x,y,z) = I(x,y,z)$. It follows that the sets
\begin{align}\label{eq:inv-surfaces}
 S_V \eqdef \set{(x,y,z)\in\R : I(x,y,z) = V},\hspace{2mm}V\in\R
\end{align}
are also preserved; moreover, since $f$ is invertible, we have, for each $V$, $f(S_V) = S_V$. In this notation, $V$ plays the role of a parameter, and $S_V$ is the surface embedded in $\R^3$ that depends on $V$, that is, $\set{S_V}_{V\in\R}$ is a one-parameter family of surfaces in $\R^3$. Notice that $V$ has also been used for a coupling constant in \eqref{eq:schrodinger}. This double use of notation is not accidental. Indeed, observe that for every $E\in\R$, $I(\gamma(E)) = \frac{V^2}{4}$. It follows that the line $\gamma$ from \eqref{eq:gamma-schro} lies on the surface $S_{\frac{V^2}{4}}$ for every choice of $V\in\R$.

Since we shall rarely have to refer to the exact value of the Fricke-Vogt invariant, let us agree on the following abuse of notation: \textit{both, the coupling constant for the Schr\"odinger operator of type \eqref{eq:schrodinger} as well as the values of the Fricke-Vogt invariant will be denoted by $V$}.

In this paper we shall be concerned only with the case $V\geq 0$ (there will be a few places where we shall have to consider $S_V$ for $V < 0$, but this is somewhat secondary to our arguments and this situation is clearly identified in the text). Indeed, observe that $I(\gamma(E))=\frac{V^2}{4}$ is clearly non-negative and is zero if and only if $V = 0$, which corresponds to the \textit{free Laplacian} case,
\begin{align*}
 (H\phi)_n = \phi_{n-1} +\phi_{n+1},
\end{align*}
the spectrum of which is the interval $[-2,2]$.

It turns out that the surfaces $S_V$ for $V > 0$ are smooth and homeomorphic to the two-dimensional sphere with four points removed, while the surface $S_0$ is smooth everywhere, except for four conic singularities. See Figure \ref{fig:surfaces} for some plots.
%
%
%
%
%
%
%
\begin{figure}[t]
	\centering
		\subfigure[$V = 0.0001$]{
			\includegraphics[scale=.3]{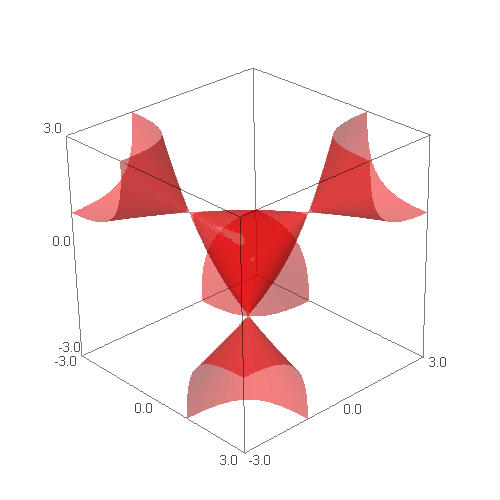}}
		\subfigure[$V = 0.01$]{
			\includegraphics[scale=.3]{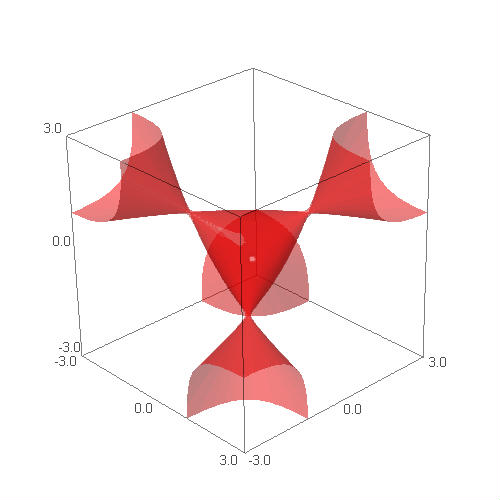}}
		\\
		\subfigure[$V = 0.05$]{
			\includegraphics[scale=.3]{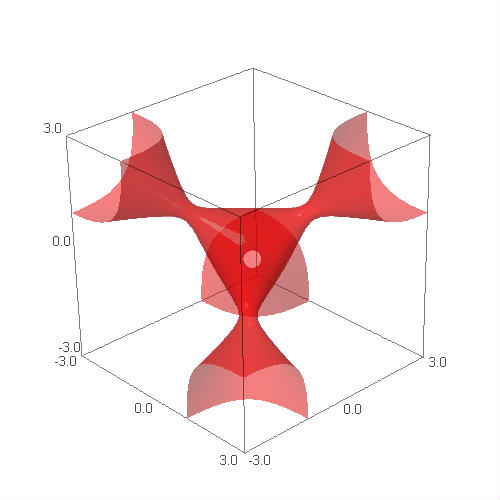}}
		\subfigure[$V = 1$]{
			\includegraphics[scale=.3]{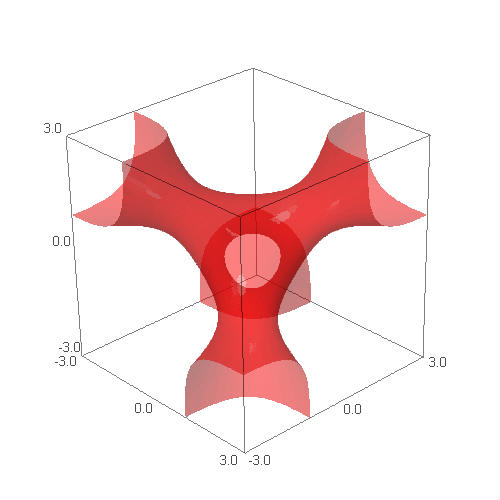}}
	\caption{Invariant surfaces $S_V$ for four values of $V$.}
	\label{fig:surfaces}
\end{figure}
%
%
%
Thus in order to understand the spectrum of the Hamiltonian in \eqref{eq:schrodinger}, it is enough to understand the structure of the set of points on $S_V$, for any $V > 0$, whose forward orbit under the action of $f$ is bounded. Before we continue on to the main results, let us introduce the following terminology.
\begin{defn}\label{defn:type-b}
For a map $f:\R^n\rightarrow \R^n$, a point $x\in\R^n$ will be called a type-\textbf{B} point if $\mathcal{O}^+_f(x)$ is bounded.
\end{defn}
We also need the notion of the \textit{dynamically defined Cantor set}, which is rigorously defined in, for example, Chapter 4 of \cite{Palis1993}. Roughly speaking, a dynamically defined Cantor set is one which is the blowup of an arbitrarily small neighborhood of any of its points under a $C^{1 + \alpha}$ map (differentiable with $\alpha$-H\"older continuous derivative). Such sets enjoy a number of nice properties, some of which are outlined below. In what follows, the Hausdorff dimension of a set $A\subset \R$ is denoted by $\hdim(A)$ and the box-counting dimension is denoted by $\bdim(A)$. For a point $a\in A$, we define the local Hausdorff dimension of $A$ at $a$ by
\begin{align*}
 \lhdim(A, a) \eqdef \lim_{\epsilon\rightarrow 0}\hdim(A\cap (a - \epsilon, a + \epsilon)),
\end{align*}
and the local box-counting dimension by
\begin{align*}
 \lbdim(A,a) \eqdef \lim_{\epsilon\rightarrow 0}\bdim(A\cap (a - \epsilon, a + \epsilon)).
\end{align*}
For a dynamically defined Cantor set $K\subset\R$, we have the following.
\it
\begin{itemize}
\setlength{\itemsep}{2pt}

 \item For every $a\in K$, $\lbdim(K, a) = \lhdim(K, a) = \bdim(K) = \hdim(K)$.

\end{itemize}
\rm

One of the consequences of M. Casdagli's contributions in \cite{Casdagli1986} is the following
\begin{thm}\label{thm:casdagli}
 For $V\geq 16$, the set of type-\textbf{B} points is formed by a smooth (so-called stable) one-dimensional lamination in $S_V$ such that any smooth and regular curve $\gamma$ in $S_V$ parameterized on a compact nonempty interval in $\R$ that is uniformly transversal to the lamination on its interior, itersects the lamination in a dynamically defined Cantor set along its interior.
\end{thm}
M. Casdagli then demonstrated that the line $\gamma$ from \eqref{eq:gamma-schro} intersects the stable lamination transversely (for all $V \geq 16$). These results have been extended to all $V$ sufficiently close to zero by D. Damanik and A. Gorodetski in \cite{Damanik2009}. As a corollary of his larger program in holomorphic dynamics, S. Cantat proved the result for all $V > 0$ in \cite{Cantat2009}. S. Cantat did not, however, prove that for all $V > 0$, the line $\gamma$ is transversal to the stable lamination; this is still open (D. Damanik and A. Gorodetski, and S. Cantat used different techniques and worked independently). We shall be referring to this result later, so for ease of reference let us state it here:

\begin{thm}\label{thm:schro-dynam-general}
 For every $V > 0$, there exists a one-dimensional smooth lamination, $\mathcal{L}_V$, in $S_V$, such that for $x\in S_V$, $\mathcal{O}_f^+(x)$ is bounded if and only if $x\in\mathcal{L}_V$. Moreover, any smooth regular curve in $S_V$, parameterized on a compact nonempty interval in $\R$, which is uniformly transversal to the lamination $\mathcal{L}_V$ on its interior intersects the lamination in a dynamically defined Cantor set.
\end{thm}

Denote the spectrum of the Hamiltonian in \eqref{eq:schrodinger} by $\sigma_V$. Observe that $\sigma_V$ depends continuously on the parameter $V$ (this is easy to see since $H$ depends continuously in norm on $V$). On the other hand, as a consequence of the above results, we in fact have analyticity of the Hausdorff dimension of the spectrum \cite{Cantat2009}:

\begin{thm}\label{thm:schro-spect-haus-cont}
 The Hausdorff dimension of $\sigma_V$, $\hdim(\sigma_V)$, is an analytic function of $V$ for $V > 0$. Moreover, for all $V > 0$, $\hdim(\sigma_V)$ takes values in the range $(0, 1)$.
\end{thm}

\begin{rem}
A few remarks are in order here. First, in order to obtain the theorem above, one needs to relate the spectrum to the lamination $\mathcal{L}_V$ from the previous theorem via the transversal intersection of the line of initial conditions, $\gamma$ from \eqref{eq:gamma-schro}, with $\mathcal{L}_V$. As we mentioned above, transversality is known for the sufficiently small and the sufficiently large values of $V$. In the intermediate regime, however, the problem is still open. On the other hand, it is known that, should transversality fail, it cannot fail at more than finitely many points, which does not introduce enough distortion to break most of the results, including the above theorem (for details, see \cite[Proof of Theorem 2.1(i)]{Yessen2011a}).

Secondly, analyticity of the Hausdorff dimension is related to analyticity of the transversal Hausdorff dimension of $\mathcal{L}_V$ (we shall discuss this later in some more detail). This follows from a theorem of R. Ma\~{n}e \cite{Mane1990}, but in lower smoothness ($C^k$, with $k\in\N$, or $C^\infty$). Ma\~{n}e's proof could be adapted to the analytic case, with some technical difficulties. On the other hand, M. Pollicott has recently provided a simpler proof of the result in \cite{Pollicott200x} (notice that Pollicott's proof is given only in the analytic category, while one still needs Ma\~{n}e's arguments in the $C^k$ or $C^\infty$ category).
\end{rem}

A relevant and challenging question is whether the Hausdorff dimension can be extended analytically to $V = 0$ from the right. This question is currently open, but \textit{continuity} from the right has been proved in \cite{Damanik2010a}:

\begin{thm}\label{thm:schro-spect-haus-cont-v0}
We have
 \begin{align*}
  \lim_{V\rightarrow 0^+}\hdim(\sigma_V) = 1.
 \end{align*}
 In fact, there exist constants $C_1, C_2 > 0$ such that
 \begin{align*}
  1 - C_1V \leq \hdim(\sigma_V)\leq 1 - C_2V
 \end{align*}
 for all $V$ sufficiently close to zero.
\end{thm}

The exact large coupling asymptotics for the Hausdorff dimension of the spectrum was given in \cite{Damanik2008}:

\begin{thm}\label{thm:degt}
We have
\begin{align*}
\lim_{V\rightarrow\infty}\hdim(\sigma_V)\cdot \log V = \log(1+\sqrt{2}).
\end{align*}
\end{thm}

\begin{rem}
Primitive noninvertible two-letter substitutions have also been considered (see, for example, \cite{Bellissard1991}); in those cases, analogues of Theorems \ref{thm:schro-spect-haus-cont} and \ref{thm:schro-spect-haus-cont-v0} are currently open questions. For extensions of Theorem \ref{thm:degt} to general Sturmian sequences, see \cite{Liu2007}.
\end{rem}

Given that the spectrum of $H$ is a Cantor set for all $V > 0$, it is natural to ask about the gap structure of the spectrum (the disjoint open subintervals that form the compliment of the spectrum). The gap labeling theorem provides a set of canonical labels for the cumulative distribution function of the spectral measure which correspond to gaps in the spectrum. See \cite{Bellissard1992} for a summary of results regarding these labelings in a more general setting. It is of interest to know whether all possible gaps that are predicted by the gap labeling are open. D. Damanik and A. Gorodetski in \cite{Damanik2010a} proved

\begin{thm}\label{thm:gap-opening}
 For all $V > 0$ sufficiently small, all gaps allowed by the gap labeling theorem are open in $\sigma_V$.
\end{thm}

\begin{rem}
 The gap structure for $H$ with $V\geq 4$ was investigated by L. Raymond in \cite{Raymond1997}. The famous “Ten Martini Problem” is showing that the spectrum of the almost Mathieu operator is a Cantor set. This was done by Avila and Jitomirskaya in 2009 \cite{Avila2009}. This is generalized to the so-called “Dry Ten Martini Problem,” showing that in fact all gaps predicted by the gap labeling theorem are open. In the above theorem, D. Damanik and A. Gorodetski solved the problem for the Hamiltonian $H$ with $V > 0$ sufficiently small.
\end{rem}

The next natural question is about the behavior of individual gaps (sizes thereof) as a function of $V$. This has also been investigated by D. Damanik and A. Gorodetski \cite{Damanik2010a}:

\begin{thm}\label{thm:gap-size}
 For the spectrum $\sigma_V$ of $H$ from \eqref{eq:schrodinger}, we have the following. Given a one-parameter family of gaps $\set{U_V}_{V > 0}$ of $\sigma_V$ (continuation of a gap along the parameter $V$), the boundary points of $U_V$ depend smoothly on $V$ for all $V > 0$ sufficiently small and, denoting the length $U_V$ by $\abs{U_V}$, we have
 \begin{align*}
  \lim_{V\rightarrow 0^+}\frac{\abs{U_V}}{V}\hspace{2mm}\text{ exists and belongs to }\hspace{2mm} (0, \infty).
 \end{align*}
\end{thm}
\begin{rem}
The limit in Theorem \ref{thm:gap-size} depends on the gap $U_V$.
\end{rem}

Besides the spectral measures, supported on the spectrum $\sigma_V$ is another natural measure: the \textit{density of states measure}. Roughly speaking, the density of states measures the amount of energy levels in a given interval that a quantum particle may occupy. This is useful in a few ways; for example, probability of occupancy of energy levels by fermions (such as electrons) can be computed in terms of the density of states. This measure is a non-atomic Borel probability measure defined by the following limiting procedure.

Let us denote by $H^{(L)}$ the restriction of the Hamiltonian $H$ to the interval $[1, L]$ (i.e. the finite lattice of nodes $1, \dots, L$ with potential $Vu_1, \dots, Vu_n$). We use the Dirichlet boundary conditions. Let us define $N(E, L)$ as the number of eigenvalues of $H^{(L)}$ that are less than or equal to $E$. Then the \textit{integrated density of states} is defined by
\begin{align}
 N(E) \eqdef \lim_{L\rightarrow\infty}\frac{1}{L}N(E, L).
\end{align}
The existence of this limit has been shown in a general setting in \cite{Lenz2005}. It has also been studied for general potentials, and analytic aperiodic potentials. See Section 5.5 in \cite{Damanik200X} for references. The density of states measure will be denoted by $dN$.

The density of sates $N$ depends on the choice of the coupling $V$. To indicate this dependence explicitly where needed, we shall write $N_V(E)$ instead of $N(E)$. The exact dimensionality of the measure $dN$ was studied by D. Damanik and A. Gorodetski in \cite{Damanik2012}, confirming the conjecture of Barry Simon that
\begin{align*}
 \inf\set{\hdim(S): N_V(S) = 1} < \hdim(\sigma_V).
\end{align*}
The exact statement of the Damanik-Gorodetski result is the following
\begin{thm}\label{thm:damanik-gorodetski-dos}
 There exists $0 < V_0 \leq \infty$ such that for all $V\in (0, V_0)$, there is $d_V\in (0, 1)$ so that the density of states measure $dN_V$ is of exact dimension $d_V$; that is, for $dN_V$-almost every $E\in\R$, we have
 \begin{align*}
  \lim_{\epsilon\downarrow 0}\frac{\log N_V(E-\epsilon, E+\epsilon)}{\log\epsilon} = d_V.
 \end{align*}
 Moreover, in $(0, V_0)$, $d_V$ is a $C^\infty$ function of $V$, and
 \begin{align*}
  \lim_{V\downarrow 0}d_V = 1.
 \end{align*}
\end{thm}

Recently M. Pollicott provided a proof of analyticity of $d_V$ in the parameter $V$ \cite{Pollicott200x}:

\begin{thm}\label{thm:idos-analiticity}
The dimension $d_V$ from Theorem \ref{thm:damanik-gorodetski-dos} depends analytically on $V$.
\end{thm}

The final result that we wish to discuss here is the set-theoretic sum of $\sigma_V$ with itself, defined as
\begin{align}\label{eq:set-sum-schro}
 \sigma^2_V \eqdef \sigma_V + \sigma_V \eqdef \set{a + b: a, b \in \sigma_V}.
\end{align}
This quantity is of interest, since it is the spectrum of the square Hamiltonian, which we denote here by $H^2$, defined as a bounded self-adjoint linear operator on $\ell^2(\Z^2)$ in the following way.
\begin{align}\label{eq:square-schro}
\begin{split}
 (H^2\phi)_{(n,m)}
 &=
 \phi_{n+1, m} + \phi_{n-1,m} + \phi_{n, m-1} + \phi_{n, m+1} + V(u_n + u_m)\phi_{n,m}.
 \end{split}
\end{align}
It follows from general principles that the spectrum of $H^2$ is precisely $\sigma^2_V$. Higher dimensional separable Hamiltonians, such as the cubic Hamiltonian, can be defined in a similar way; the spectrum of the cubic one, for example, is $\sigma^2_V + \sigma_V$. Since for $V > 0$, $\sigma_V$ is a Cantor set, it is a nontrivial problem to describe the topology of $\sigma^2_V$. There is, however, a sufficient condition which insures that $\sigma^2_V$ is an interval (which holds for $V > 0$ sufficiently small). It is also known that for all $V$ sufficiently large, $\sigma^2_V$ is a Cantor set (the intermediate case is an interesting and a challenging problem). Before we can state this sufficient condition, we need the notion of \textit{thickness} of a Cantor set, which we define next.

Let $K\subset\R$ be a Cantor set and denote by $\mathcal{I}$ its convex hull. Any connected component of $\mathcal{I}\setminus K$ is called a \textit{gap} of $K$. A \textit{presentation} of $K$ is given by an ordering $\mathcal{U} = \set{U_n}_{n\geq 1}$ of the gaps of $K$. If $u\in K$ is a boundary point of a gap $U$ of $K$, we denote by $K_u$ the connected component of $\mathcal{I}\setminus(U_1\cup U_2\cup\cdots\cup U_n)$ (with $n$ chosen so that $U_n = U$) that contains $u$, and write

\begin{align*}
\tau(K_u,\mathcal{U},u) = \frac{\abs{K_u}}{\abs{U}}.
\end{align*}

With this notation, the \textit{thickness} $\tau(K)$ of $K$ is given by

\begin{align}\label{eq:thickness}
\tau(K) = \sup_{\mathcal{U}}\inf_u\tau(K_u,\mathcal{U},u).
\end{align}

Next we define what S. Astels in \cite{Astels2000} calls \textit{normalized thickness}. For a Cantor set $K$, define the normalized thickness of $K$ by
\begin{align*}
\bar{\tau}(K) = \frac{\tau(K)}{1 + \tau(K)}
\end{align*}
(Astels originally used the letter $\gamma$ for normalized thickness, which we use here for the line of initial conditions to preserve notation from \cite{Yessen2011, Yessen2011a} for easy reference).

A specialized version of the more general \cite[Theorem 2.4]{Astels2000} gives the following remarkable result.

\begin{thm}\label{thm:thm-astels}
Given a Cantor set $K$, if $m\bar{\tau}(K)\geq 1$, $m\in\N$, then the $m$-fold arithmetic sum of $K$ with itself is equal to the $m$-fold sum of $\mathcal{I}$ with itself, where $\mathcal{I}$ is the hull of $K$.
\end{thm}

We should remark that the result of Astels is a (nontrivial) generalization of the result which is better known in the theory of hyperbolic dynamical systems as (a consequence of) S. Newhouse's \textit{gap lemma} \cite{Newhouse1968}:

\begin{thm}\label{thm:gap-lemma}
 If $K_1, K_2$ are Cantor sets and $\tau(K_1)\tau(K_2) > 1$, then either one of the Cantor sets lies in the gap of the other (or outside of the other's hull), or $K_1\cap K_2\neq \infty$.
\end{thm}
Then we have the following corollary (of which Theorem \ref{thm:thm-astels} is a generalization).
\begin{cor}\label{cor:newhouse}
 If $K_1, K_2$ are Cantor sets such that one intersects the hull of the other but does not fall entirely into a gap of the other, and $\tau(K_1)\tau(K_2) > 1$, then $K_1 + K_2$ is an interval.
\end{cor}
D. Damanik and A. Gorodetski proved the following theorem in \cite{Damanik2010a}.

\begin{thm}\label{thm:thickness-schro}
 $\lim_{V\rightarrow 0}\tau(\sigma_V) = \infty$ and $\lim_{V \rightarrow\infty}\tau(\sigma_V) = 0. $
\end{thm}

It is also known that if a Cantor set has Hausdorff dimension sufficiently small, then its sum with itself is a Cantor set (see \cite[Chapter 4]{Palis1993}). Thus the following is an immediate corollary of Theorems \ref{thm:thickness-schro} and \ref{thm:degt}.

\begin{cor}\label{cor:spectrum-sum}
There exist $V_0, V_1 > 0$ such that the spectrum of the square Fibonacci Hamiltonian is an interval for $V\in (0, V_0)$, and a Cantor set for $V\in(V_1, \infty)$.
\end{cor}

\begin{rem}
As we have already mentioned above, it is currently unknown what happens in the intermediate coupling regime. Work in this direction is in progress and some results have already appeared in \cite{Damanik2013X}.
\end{rem}

Let us mention a relationship between the Hausdorff dimension and the thickness of a Cantor set (for details, see Chapter 4 in \cite{Palis1993}):
\begin{align*}
 \hdim(K) \geq \frac{\log 2}{\log \left(2 + \frac{1}{\tau(K)}\right)}.
\end{align*}
In fact, Theorem \ref{thm:schro-spect-haus-cont-v0} was proved in \cite{Damanik2010a} by proving estimates on $\tau(\sigma_V)$, and, in the other direction, by giving estimates on the \textit{denseness} (which bounds the Hausdorff dimension in a way analogous to that of thickness, from the other side, and which is defined analogously to thickness by replacing $\sup$ and $\inf$ in \eqref{eq:thickness} by $\inf$ and $\sup$, respectively). For further details and properties of different quantitative characteristics of (dynamically defined, in particular) Cantor sets, see Chapter 4 in \cite{Palis1993}.

The program that has been carried out for the operator \eqref{eq:schrodinger} by D. Damanik and A. Gorodetski and colleagues is extensive. The proofs of all of the above results rely heavily on dynamics of the Fibonacci trace map. The most recent results, as well as some open problems, are collected in the survey \cite{Damanik2012}, with complementing numerical analysis.

Let us now discuss some recent extensions of the above results to the more general class of potentials generated by primitive invertible substitutions on two letters and to the tridiagonal Fibonacci Hamiltonian.

\subsection{The primitive invertible Schr\"odinger case}\label{sec:generalization-mei}

The purpose of this section is to present recent extensions of some of the results from the previous section to the general case of potentials modulated by primitive invertible two-letter substitutions. The details are presented in \cite{Mei2013}. Let us begin by fixing notation.

For any primitive invertible two-letter substitution sequence $v$ and $\omega\in\Omega_v$, the hull of $v$ from \eqref{eq:hull}, we shall denote by $H_{\omega, V}$, $V > 0$, the operator \eqref{eq:schrodinger} with the potential given by $V_n = V\omega_n$.

The first fundamental result regarding the spectra of operators $H_{\omega, V}$ is the following theorem, due to D. Damanik \cite{Damanik2000}.

\begin{thm}\label{thm:damanik}
 Given any primitive invertible two-letter substitution sequence $v$ and $V > 0$, for any $\omega\in\Omega_v$, the spectrum of $H_{\omega, V}$ coincides with that of $H_{v, V}$. The spectrum is a Cantor set of zero Lebesgue measure and of purely singular continuous type.
\end{thm}

More generally, a large class of Schr\"odinger operators is known to have purely singular continuous spectrum, see \cite{Lenz2002} and \cite{LiuTanWenWu2002}. As we have already mentioned above, we may associate a trace map to every primitive invertible substitution which recursively expresses the traces of products of two matrices that follow (in the product) the same order as the letters of the sequence. These maps are necessarily polynomial, are invertible, and preserve the Fricke-Vogt invariant (and hence the surfaces $S_V$) from equation \eqref{eq:fv-invariant} (and equation \eqref{eq:inv-surfaces}). We give more details in Section \ref{sec:proofs}. Recall from above that the theorem of A. S\"ut\H{o}, Theorem \ref{thm:suto}, related the spectrum of $H_{u,V}$ to the set of type-\textbf{B} points for the Fibonacci trace map $f$. It turns out that a similar relation exists for the general primitive invertible case, as proved in \cite{Damanik2000}.

Indeed, for any primitive invertible two-letter substitution $v$, there exists a line $\gamma_v: \R\rightarrow \R^3$ and a polynomial diffeomorphism $f_v: \R^3\rightarrow\R^3$ such that $E$ is in the spectrum of $H_{v, V}$ if and only if $\mathcal{O}_{f_v}^+(\gamma(E))$ is bounded (more details are given, and the line $\gamma_v(E)$ is constructed, in fact in the more general case of Jacobi operators, in Section \ref{sec:proofs}). In fact, $\gamma_v$ is equal to $\gamma$ from \eqref{eq:gamma-schro} for primitive invertible $v$. The following extension of Theorem \ref{thm:schro-dynam-general} was carried out in \cite{Mei2013}.

\begin{thm}\label{thm:mei1}
 For any primitive invertible two-letter substitution sequence $v$ and every $V > 0$, there exists a one-dimensional smooth lamination, $\mathcal{L}_{v,V}$, in $S_V$, such that for $x\in S_V$, $\mathcal{O}_{f_v}^+(x)$ is bounded if and only if $x\in\mathcal{L}_{v,V}$. Moreover, any smooth regular curve in $S_V$, parameterized on an open interval in $\R$, that is transversal to the lamination $\mathcal{L}_V$ intersects the lamination in a dynamically defined Cantor set.
\end{thm}

Theorem \ref{thm:mei1} allowed us to obtain the following series of results as extensions from the Fibonacci case to the general primitive invertible case.

\begin{thm}\label{thm:mei2}
 For any primitive invertible two-letter substitution $v$, Theorem \ref{thm:schro-spect-haus-cont-v0} holds with $H$ replaced by $H_{v,V}$.
\end{thm}

\begin{thm}\label{thm:mei3}
 For any primitive invertible two-letter substitution $v$, Theorem \ref{thm:gap-opening} holds with $H$ replaced by $H_{v, V}$.
\end{thm}

\begin{thm}\label{thm:mei4}
 For any primitive invertible two-letter substitution $v$, Theorem \ref{thm:gap-size} holds with $H$ replaced by $H_{v,V}$.
\end{thm}

\begin{thm}\label{thm:mei5}
 For any primitive invertible two-letter substitution $v$, Theorem \ref{thm:damanik-gorodetski-dos} holds with $H$ replaced by $H_{v,V}$.
\end{thm}

(We chose to state the four theorems above separately for the ease of reference later in this paper.)

\subsection{The Fibonacci Jacobi case}\label{sec:generalization-yessen}

In this section we are concerned with the spectrum of the Hamiltonian \eqref{eq:gen-jacobi} where
\begin{align*}
 a_n = p(u_n)\hspace{2mm}\text{ and }\hspace{2mm}b_n = q(u_n),
\end{align*}
with $p:\set{0,1}\rightarrow\R_{\neq 0}$ and $q:\set{0,1}\rightarrow\R$, and $\set{u_n}$ is the Fibonacci substitution sequence. Independently of the choice of the sequences $\set{a_n}$ and $\set{b_n}$, the operator is clearly self-adjoint, and bounded if the sequences $\set{a_n}$ and $\set{b_n}$ are bounded. Since the proofs of our main results in Section \ref{sec:proofs} will rely heavily on the transfer matrix and the trace map formalism, let us present the formalism here, before we review some of the recent results. Throughout this section, $H$ will denote the operator \eqref{eq:gen-jacobi}.

Observe that $\theta \in \R^\Z$ solves the equation
\begin{align}\label{eq:eigenvalue-jacobi}
 H\theta = E\theta
\end{align}
if and only if
\begin{align}\label{eq:jacobi-transfer}
 \begin{pmatrix}
  \theta_{n+1}\\
  \theta_{n}
 \end{pmatrix}
 =
 M_n(E)
 \begin{pmatrix}
  \theta_{n}\\
  \theta_{n-1}
 \end{pmatrix},
\end{align}
where the matrix $M_n$ is now redefined as
\begin{align}\label{eq:j-transfer-one-site}
 M_n(E) \eqdef \frac{1}{a_{n+1}}
 \begin{pmatrix}
  E - b_{n}	&	-a_{n}\\
  a_{n+1}		&	0	\\
 \end{pmatrix}.
\end{align}
We of course assume that $a_n\neq 0$; hence the assumption that $p$ maps $\set{0,1}$ into $\R_{\neq 0}$. From this we get
\begin{align*}
 \begin{pmatrix}
  \theta_{n+1}\\
  \theta_n
 \end{pmatrix}
 =
 M_n(E)M_{n-1}(E)\cdots M_1(E)
 \begin{pmatrix}
  \theta_1\\
  \theta_0
 \end{pmatrix}.
\end{align*}
As in the Schr\"odinger case, Floquet theory also applies to the Jacobi case. We do not elaborate on this much further here (for details, see Section 3 in \cite{Yessen2011a}). Notice, however, that unlike in the Schr\"odinger case, the matrices $M_n(E)$ are not necessarily unimodular. This can be remedied as follows. Define
\begin{align*}
 T_n(E) \eqdef \frac{1}{a_{n+1}}
 \begin{pmatrix}
  E - b_n	&	-1\\
  a_{n+1}^2		&	0
 \end{pmatrix}.
\end{align*}
Notice that $T_n(E)$ is unimodular and equation \eqref{eq:jacobi-transfer} holds if and only if
\begin{align*}
 \Theta_n = T_n(E)\Theta_{n-1}
\end{align*}
where $\Theta_n = (\theta_n, a_n\theta_{n-1})$. Now define $\hat{T}_n(E) \eqdef T_n(E)T_{n-1}(E)\cdots T_1(E)$. Since $\hat{T}_n(E)$ is unimodular, and the sequences $\set{a_n}$ and $\set{b_n}$ are modulated by the Fibonacci substitution, it is easy to see that, much like in the Schr\"odinger case, the matrices $\hat{T}_n(E)$ follow the same recursion from right to left:
\begin{align*}
 \hat{T}_{F_{k+1}}(E) = \hat{T}_{F_{k-1}}(E)\hat{T}_{F_k}(E).
\end{align*}
At this point the Fibonacci trace map can be applied. Again, with $x_k = x_k(E) \eqdef \frac{1}{2}\Tr \hat{T}_k(E)$, we have
\begin{align*}
 (x_{k+2}, x_{k+1}, x_k) = f(x_{k+1}, x_k, x_{k-1}),
\end{align*}
with $f$ as in \eqref{eq:fib-tmap}.

Notice that after shifting and scaling the operator appropriately, we may assume that $p(0) = 1$ and $q(0) = 0$. In \cite{Yessen2011a} we introduced the following shorthand notation, which we also adopt here: with $p(1) = \mathfrak{p}$ and $q(1) = \mathfrak{q}$, let $H_{(\mathfrak{p},\mathfrak{q})}$ denote the operator \eqref{eq:gen-jacobi} with $a_n = p(u_n)$ and $b_n = q(u_n)$ ($u_n$ being the Fibonacci substitution sequence) and the additional assumption that $p(0) = 1$ and $q(0) = 0$.

From \cite{Yessen2011a}, we have the following fundamental result, which is in line with that for the Schr\"odinger operator (see Theorem \ref{thm:damanik} in Section \ref{sec:generalization-mei} above).

\begin{thm}\label{thm:fib-jacobi-sc-spectr}
 For $\omega\in\Omega_u$, the hull of the Fibonacci substitution sequence $u$, let $H_{\omega, (\mathfrak{p}, \mathfrak{q})}$ be defined as $H_{(\mathfrak{p},\mathfrak{q})}$ with $\omega$ in place of $u$. Then for all $(\mathfrak{p},\mathfrak{q})\neq (1,0)$, the spectrum of $H_{\omega, (\mathfrak{p},\mathfrak{q})}$ is independent of the choice of $\omega$, is a zero measure Cantor set, and is of purely singular continuous type.
\end{thm}

\begin{rem}\label{rem:beckus}
Let us remark that the more general case of Jacobi operators with Sturmian disorder has recently been studied in \cite{Beckus2013}. Spectral and ergodic theory were used there. The finer structure of the fractal spectrum (such as fractal dimensions) were not investigated there. Let us also remark that \cite{Beckus2013} came to our attention after the first draft of the present work had been completed and submitted for publication; thus the results of the present paper and \cite{Beckus2013} were obtained independently.
\end{rem}

Following the technique that has been applied for the Schr\"odinger operator, in \cite{Yessen2011a} we proved

\begin{prop}\label{prop:jacobi-tmap}
 The point $E\in\R$ is in the spectrum of $H_{(\mathfrak{p},\mathfrak{q})}$ if and only if $\mathcal{O}_f^+(\gamma(E))$ is bounded, where
 \begin{align}\label{eq:gamma-fib-jacobi}
  \gamma(E) = \left(\frac{E - \mathfrak{q}}{2}, \frac{E}{2\mathfrak{p}}, \frac{1 + \mathfrak{p}^2}{2\mathfrak{p}}\right).
 \end{align}
\end{prop}

Now that we have a criterion in terms of the dynamics of $f$ for points in the spectrum, one would imagine that the same techniques of dynamical systems that were applied in the Schr\"odinger case could be applied here to extend the previous results from the Schr\"odinger to the Jacobi case. This is in principle true, and this is the route that we followed in \cite{Yessen2011a}; however, the previous techniques do not generalize \textit{verbatim} and present some technical difficulties.

Indeed, observe that
\begin{align*}
 I(\gamma(E)) = \frac{E\mathfrak{q}(1-\mathfrak{p}^2) + \mathfrak{q}^2\mathfrak{p}^2 + (\mathfrak{p}^2 - 1)^2}{4\mathfrak{p}^2}.
\end{align*}
We thus have
\it
\begin{itemize}
\setlength{\itemsep}{2pt}

 \item $I(\gamma(E))$ is $E$-dependent when $\mathfrak{p}\neq 1$ and $\mathfrak{q}\neq 0$ (in this case, in fact, the line $\gamma$ intersects the surface $S_V$ from \eqref{eq:inv-surfaces}, for $V\in\R$, in at most one point).

 \item Moreover, there exist values of $E$ for which $I(\gamma(E)) < 0$.
\end{itemize}
\rm
Both points above stand in contrast with the Schr\"odinger situation, where we have
\it
\begin{itemize}
\setlength{\itemsep}{2pt}

 \item $I(\gamma(E))$ is $E$-independent.

 \item $I(\gamma(E)) \geq 0$ for all $E$ and is equal to zero if and only if $V = 0$,
\end{itemize}
\rm
with $\gamma(E)$ as in \eqref{eq:gamma-schro} and $V$ as in \eqref{eq:schrodinger}. This in turn allows one to use Theorem \ref{thm:schro-dynam-general} to investigate fractal properties of the spectrum (at least in the case that $\gamma$ is transversal to the lamination $\mathcal{L}$ from Theorem \ref{thm:schro-dynam-general}). Thus in order to extend the results to the Jacobi case, we had to combat two technical difficulties (now with $\gamma$ from \eqref{eq:gamma-fib-jacobi}):
\it
\begin{itemize}
\setlength{\itemsep}{2pt}

 \item For the portion of $\gamma$ that lies in $\cup_{V > 0}S_V$, investigate the intersection with $\cup_{V > 0}\mathcal{L}_V$. For this purpose, describe the geometry of $\cup_{V > 0}\mathcal{L}_V$ (if any).

 \item Investigate the intersection of $\gamma$ with $\cup_{V < 0}S_V$ and any type-\textbf{B} points in this region.
\end{itemize}
\rm
The second point is actually easy as a consequence of some results from \cite{Roberts1996} (which we do not discuss in detail here, but postpone until Section \ref{sec:proofs}): the portion of $\gamma$ that lies in $\left(\cup_{V < 0}S_V\right)$ does not contain any type-\textbf{B} points.

The first point was attacked in our previous work \cite{Yessen2011}, where we showed that
\begin{align}\label{eq:2d-lamination}
 \mathcal{L}\eqdef\bigcup_{V > 0}\mathcal{L}_V
\end{align}
forms a two-dimensional lamination with smooth leaves in
\begin{align}\label{eq:3d-manifold}
 \mathcal{M}\eqdef \bigcup_{V > 0}S_V,
\end{align}
and in \cite{Yessen2011a} we showed that generally $\gamma$ intersects $\mathcal{L}$ transversely everywhere except (possibly) for finitely many points, and transversely wherever $(\mathfrak{p}, \mathfrak{q})$ is sufficiently close to $(1,0)$ (this is the condition analogous to $V$ being sufficiently close to zero in the Schr\"odinger case). As a consequence, we were able to prove the following series of theorems in \cite{Yessen2011a}.

In all of what follows, denote by $\Sigma_{(\mathfrak{p},\mathfrak{q})}$ the spectrum of $H_{(\mathfrak{p},\mathfrak{q})}$.

\begin{thm}\label{thm:yessen1}
 For all $(\mathfrak{p},\mathfrak{q})\neq (1,0)$, $\Sigma_{(\mathfrak{p},\mathfrak{q})}$ is a multifractal; more precisely, the following holds.
 \begin{enumerate}[(1)]
 \setlength{\itemsep}{2pt}

  \item $\lhdim(\Sigma_{(\mathfrak{p},\mathfrak{q})}, a)$, as a function of $a\in \Sigma_{(\mathfrak{p},\mathfrak{q})}$, is continuous; it is constant if $\mathfrak{p} = 1$ or $\mathfrak{q} = 0$, and non-constant otherwise.

  \item There exists a nonempty set $\mathfrak{N}\subset\R^2$ of Lebesgue measure zero, such that the following holds.
  \begin{enumerate}[(a)]
  \setlength{\itemsep}{2pt}

   \item For all $(\mathfrak{p},\mathfrak{q})\notin\mathfrak{N}$, we have $0 < \lhdim(\Sigma_{(\mathfrak{p},\mathfrak{q})}, a) < 1$ for all $a\in \Sigma_{(\mathfrak{p},\mathfrak{q})}$; hence we have $0 < \hdim(\Sigma_{(\mathfrak{p},\mathfrak{q})}) < 1$.

   \item For $(\mathfrak{p},\mathfrak{q})\in\mathfrak{N}$, $0 < \lhdim(\Sigma_{(\mathfrak{p},\mathfrak{q})}, a) < 1$ for all $a\in \Sigma_{(\mathfrak{p},\mathfrak{q})}$ away from the lower and upper boundary points of the spectrum, and $\hdim(\Sigma_{(\mathfrak{p},\mathfrak{q})}) = 1$. In fact, the dimension accumulates at one of the endpoints of the spectrum.
  \end{enumerate}

  \item $\lim_{(\mathfrak{p},\mathfrak{q})\rightarrow(1,0)}\hdim(\Sigma_{(\mathfrak{p},\mathfrak{q})}) = 1$. In fact, the Hausdorff dimension of the spectrum is a continuous function of the parameters.

  \item $\hdim(\Sigma_{(\mathfrak{p},0)})$ and $\hdim(\Sigma_{(1,\mathfrak{q})})$ depend analytically on $\mathfrak{p}$ and $\mathfrak{q}$, respectively.
 \end{enumerate}
\end{thm}

\begin{rem}
 We remark that the case of $\mathfrak{q} = 0$, the so called \textit{off-diagonal model}, has been studied in \cite[Appendix A]{Damanik2010a} as well as \cite{Dahl2010}. The arguments relying on the dynamics of $f$ are essentially the same in the off-diagonal case as in the Schr\"odinger case. Notice, however, that (2)(a) and (4) above are extensions of the previous results, since the previous results relied on transversality conditions (hence were proved only for $\mathfrak{p}$ sufficiently close to one and for $\mathfrak{q}$ sufficiently small or sufficiently large).
\end{rem}

To aid with visualization of the next result, refer to Figure \ref{fig:regions}.

\begin{thm}\label{thm:yessen2}
 The following statements hold.
 \begin{enumerate}[(1)]
 \setlength{\itemsep}{2pt}

  \item There exists $\epsilon > 0$ such that for all $(\mathfrak{p}, \mathfrak{q})$ satisfying $\norm{(1,0) - (\mathfrak{p},\mathfrak{q})} < \epsilon$, the box-counting dimension of $\Sigma_{(\mathfrak{p},\mathfrak{q})}$ exists and coincides with the Hausdorff dimension.

  \item There exists $\Delta > 0$, such that for all $\abs{\mathfrak{p}}\geq \Delta$, there exists $\delta_{\mathfrak{p}} > 0$, so that for all $\abs{\mathfrak{q}} < \delta_{\mathfrak{p}}$, the box-counting dimension of $\Sigma_{(\mathfrak{p},\mathfrak{q})}$ exists and coincides with the Hausdorff dimension.

  \item There exists $\Delta > 0$, such that for all $\abs{\mathfrak{q}}\geq \Delta$, there exists $\delta_{\mathfrak{q}} > 0$, so that for all $\abs{\mathfrak{p}-1} < \delta_{\mathfrak{q}}$, the box-counting dimension of $\Sigma_{(\mathfrak{p},\mathfrak{q})}$ exists and coincides with the Hausdorff dimension.
 \end{enumerate}
\end{thm}

\begin{rem}\label{rem:regions}
 The reason for the restrictions in the previous theorem, is that in the arguments for the existence of the box-counting dimension and its equality with the Hausdorff dimension we employ the transversality condition, which we can prove only in some special cases (but we conjecture that it holds in general).
\end{rem}
\begin{figure}[t]
 \includegraphics[scale=.6]{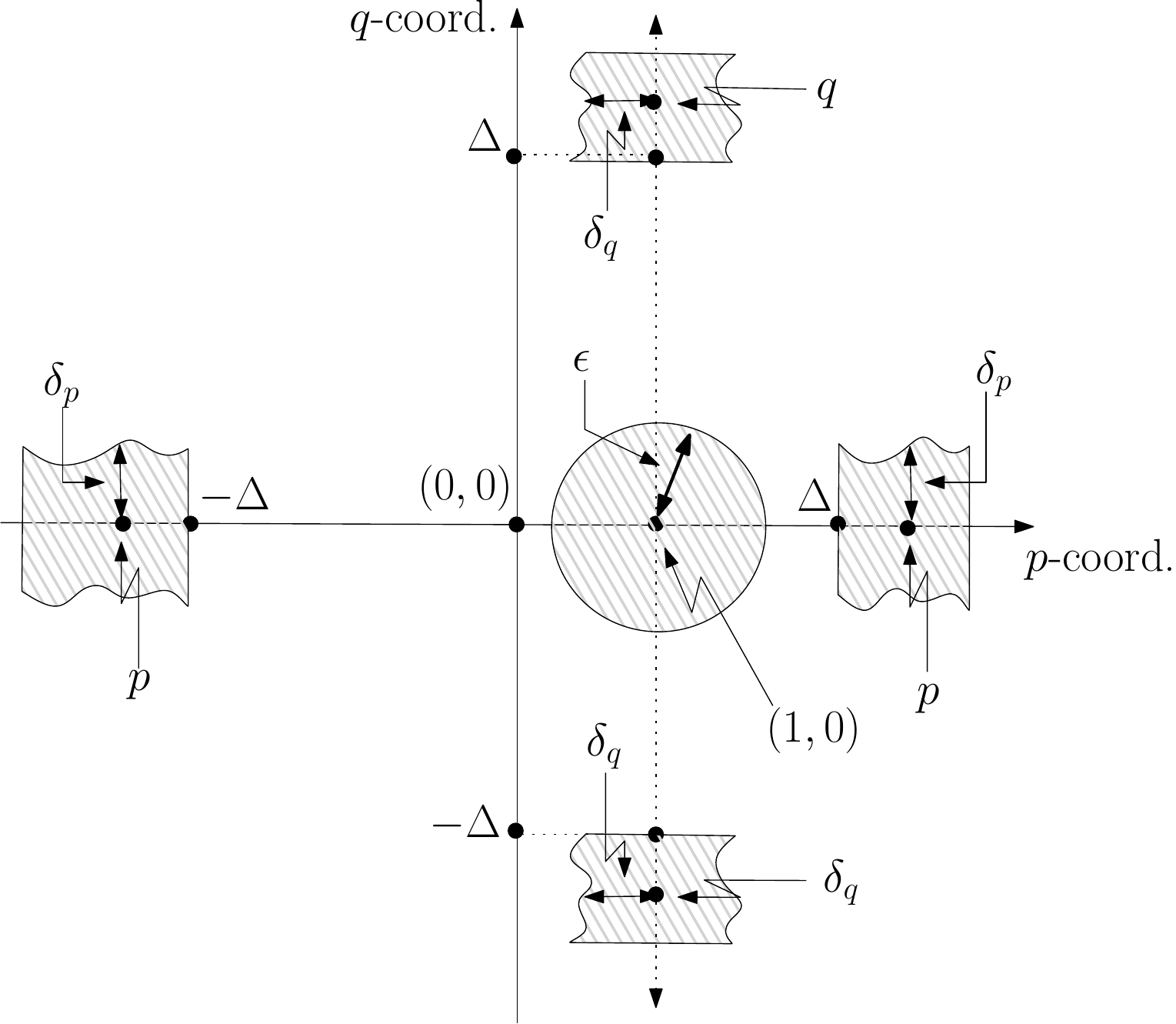}
 \caption{}
\end{figure}\label{fig:regions}
Just as in the Schr\"odinger case, denote the density of states by $N$ and the corresponding measure by $dN$ (for properties and examples see, for instance, \cite[Chapter 5]{Teschl1999}).

\begin{thm}\label{thm:yessen3}
 For all $(\mathfrak{p},\mathfrak{q})\in\R^2$, there exists $\mathfrak{D}_{(\mathfrak{p},\mathfrak{q})}\subset\R$ of full $dN$ measure, such that for all $E\in \mathfrak{D}_{(\mathfrak{p},\mathfrak{q})}$, we have
 \begin{align*}
  \lim_{\epsilon\downarrow 0}\frac{\log N(E - \epsilon, E+\epsilon)}{\log\epsilon} = d_{(\mathfrak{p},\mathfrak{q})}(E)\in\R,
 \end{align*}
 $d_{(\mathfrak{p},\mathfrak{q})}(E) > 0$. Moreover, if $(\mathfrak{p},\mathfrak{q})\neq(1,0)$, then
 \begin{align*}
  d_{(\mathfrak{p},\mathfrak{q})}(E) < \lhdim(\Sigma_{(\mathfrak{p},\mathfrak{q})}, E).
 \end{align*}
 Also,
 \begin{align*}
  \lim_{(\mathfrak{p},\mathfrak{q})\rightarrow(1,0)}\sup_{E\in\mathfrak{D}_{(\mathfrak{p},\mathfrak{q})}}\set{d_{(\mathfrak{p},\mathfrak{q})}(E)} = \lim_{(\mathfrak{p},\mathfrak{q})\rightarrow(1,0)}\inf_{E\in\mathfrak{D}_{(\mathfrak{p},\mathfrak{q})}}\set{d_{(\mathfrak{p},\mathfrak{q})}(E)} = 1.
 \end{align*}
\end{thm}

\section{New results: the primitive invertible Jacobi case.}\label{sec:main}

The purpose of this section is to unify the results of Sections \ref{sec:generalization-mei} and \ref{sec:generalization-yessen}, as well as to extend those results further. Let us begin by fixing the notation.

For a primitive invertible substitution sequence $v$ of two letters $\set{0,1}$, and $\omega\in\Omega_v$, we shall denote by $H_{\omega, (\mathfrak{p},\mathfrak{q})}$ the operator on $\ell^2(\Z)$ defined by
\begin{align}\label{eq:full-jacobi}
 (H_{\omega, (\mathfrak{p}, \mathfrak{q})}\phi)_n = p(\omega_{n+1})\phi_{n+1} + p(\omega_n)\phi_{n-1} + q(\omega_n)\phi_n,
\end{align}
where $(p(0), p(1)) = (1,\mathfrak{p})$, with $\mathfrak{p}\in\R_{\neq 0}$, and $(q(0),q(1)) = (0, \mathfrak{q})$ for any $\mathfrak{q}\in\R$. Notice that these restrictions do not affect the generality of our results, since for a general choice of $(p(0), p(1))\in\R^2$, $p(i)\neq 0$, and $(q(0),q(1))\in\R^2$, the operator in \eqref{eq:full-jacobi} is equivalent to the one with the above restrictions via an appropriate shift and scaling.

We shall denote the spectrum of the operator \eqref{eq:full-jacobi} by $\Sigma_{\omega, (\mathfrak{p},\mathfrak{q})}$, and when there is no risk of ambiguity, we shall drop $\omega$ from the notation and write simply $\Sigma_{(\mathfrak{p}, \mathfrak{q})}$.

We are interested in the topology of $\Sigma_{\omega, (\mathfrak{p},\mathfrak{q})}$ as well as the spectral type. The first fundamental result in this direction, which is an extension of the Schr\"odinger case from \cite{Damanik2000}, is the following

\begin{thm}\label{thm:mei-yessen-1}
For every primitive invertible substitution sequence $v$ and $(\mathfrak{p},\mathfrak{q})\neq (1,0)$, there exists $\Sigma_{v, (\mathfrak{p},\mathfrak{q})}\subset\R$ such that for every $\omega\in\Omega_v$, $\Sigma_{\omega, (\mathfrak{p},\mathfrak{q})} = \Sigma_{v, (\mathfrak{p},\mathfrak{q})}$; moreover, $\Sigma_{v, (\mathfrak{p},\mathfrak{q})}$ is a Cantor set of zero Lebesgue measure. The spectral measures are purely singular continuous.
\end{thm}

\begin{rem}
Remark \ref{rem:beckus} is also in order here.
\end{rem}

Next we give a criterion for points in the spectrum in terms of dynamics of the associated trace map. Given a substitution sequence $v$, let us consider the following periodic approximation for $v$. Given that $v$ is the substitution sequence generated by, say, a substitution $s$ (and let us assume without loss of generality that $s(0)$ begins with $0$), define $v^{(k)}$ to be the infinite periodic sequence of period $|s^k(0)|$, the length of the word $s^k(0)$, constructed by repeating the word $s^k(0)$. Let us call the corresponding operator, with $v^{(k)}$ in place of $v$, $H^{(k)}_{v, (\mathfrak{p}, \mathfrak{q})}$. Notice that as $k\rightarrow\infty$, we have $H^{(k)}_{v,(\mathfrak{p},\mathfrak{q})}\rightarrow H_{v, (\mathfrak{p},\mathfrak{q})}$ strongly. It is natural to ask how the spectra, $\sigma_k$, of $H^{(k)}_{v,(\mathfrak{p}, \mathfrak{q})}$ approximate the spectrum $\sigma$ of $H_{v, (\mathfrak{p},\mathfrak{q})}$. In general, of course, strong convergence does not guarantee convergence of spectra; in our case, however, we have the following
\begin{thm}\label{thm:mei-yessen-2a}
With the above setup, $\sigma_k\rightarrow\sigma$ in the Hausdorff metric as $k\rightarrow\infty$.
\end{thm}
\begin{rem}
The previous theorem can be proved without using the trace map formalism, but only spectral analytic methods. Indeed, this has been done for the almost Mathieu operator in \cite[Theorem 7.3]{Avron1990}, and also in a more general setting in \cite[Section 4]{Choi1990}. Thus our result is only a special case that should be viewed as a compliment of our trace map methods. Let us also remark that this method of proving continuity of spectra was discovered in \cite{Yessen2012a} while investigating the zero set of the partition function in the thermodynamic limit of the one-dimensional Fibonacci Ising model, where, at least at the first glance, spectral methods were not applicable. Interestingly, a duality between the 1D classical Ising models and a class of unitary operators, the so-called CMV matrices, was established later in \cite{Damanik2013b}, which now allows for applications of spectral methods.
\end{rem}
Let us return to the density of states. To indicate dependence on the primitive invertible substitution sequence $v$ and on the parameters $(\mathfrak{p},\mathfrak{q})$, let us write $N_{v, (\mathfrak{p},\mathfrak{q})}$. Observe that by construction, the density of states is nondecreasing and is constant on the gaps of the spectrum. Since the spectrum is a Cantor set, it has countably many gaps, and these gaps are ``detected'' by the density of states as regions on which $N_{v, (\mathfrak{p},\mathfrak{q})}$ is constant. Thus the gaps may be labeled by the values of $N_{v, (\mathfrak{p},\mathfrak{q})}$ that remain constant on a sufficiently small neighborhood. In \cite{Bellissard1992}, a general gap labeling scheme was derived for Schr\"odinger operators whose potential is given by a rotation sequence. In our case, it can be written down as follows. Suppose $v$ is a primitive invertible substitution sequence, and $\alpha_v\in (0,1)$ is an irrational number such that $v$ can be obtained by sampling an orbit under the irrational circle rotation by $\alpha_v$ as in \eqref{eq:sampling}.
Then, letting $\set{x}$ denote $x\mod 1$, we have
\begin{align}
 \set{N_{v, (\mathfrak{p},\mathfrak{q})}(E): E\in\R\setminus\Sigma_{v,(\mathfrak{p},\mathfrak{q})}}\subseteq \set{\set{m\alpha_v}: m\in\Z}.
\end{align}
\begin{thm}\label{thm:mei-yessen-3}
 For every primitive invertible substitution sequence $v$ there exists $\delta > 0$ such that for all $(\mathfrak{p},\mathfrak{q})$ with $0<\norm{(\mathfrak{p},\mathfrak{q}) - (1,0)} < \delta$, all gaps allowed by the gap labeling theorem are open in $\Sigma_{v, (\mathfrak{p},\mathfrak{q})}$; that is, we have
\begin{align*}
\set{N_{v, (\mathfrak{p},\mathfrak{q})}(E): E\in\R\setminus\Sigma_{v,(\mathfrak{p},\mathfrak{q})}}= \set{\set{m\alpha_v}: m\in\Z}.
\end{align*}
\end{thm}

\begin{rem}
We conjecture that the weak-coupling restriction (that is, the $\delta$ in the theorem above) can be removed.
\end{rem}

\begin{thm}\label{thm:mei-yessen-4}
 The function $\hdim(\Sigma_{v, (\mathfrak{p},\mathfrak{q})})$ is continuous in $(\mathfrak{p},\mathfrak{q})$; moreover, $\hdim(\Sigma_{v, (\mathfrak{p}, 0)})$ and $\hdim(\Sigma_{v, (1,\mathfrak{q})})$ are analytic in $\mathfrak{p}$ and $\mathfrak{q}$, respectively, for $\mathfrak{p}\neq 1$ and $\mathfrak{q}\neq 0$. Moreover, for all $(\mathfrak{p}, \mathfrak{q})$ sufficiently close to $(0,1)$, there exist constants $C_1, C_2 > 0$ such that
 \begin{align}\label{eq:y1}
  1 - C_1\norm{(\mathfrak{p}, \mathfrak{q}) - (1, 0)} \leq \hdim(\Sigma_{v, (\mathfrak{p},\mathfrak{q})}) \leq 1.
 \end{align}
\end{thm}

\begin{rem}
 We conjecture analyticity in $(\mathfrak{p},\mathfrak{q})$ (away from $(1,0)$), rather than only continuity. Another conjecture states that the Hausdorff dimension of the spectrum of the Schr\"odinger operator with Fibonacci potential is a monotone function of the coupling constant. To see how our analyticity conjecture follows from the monotonicity conjecture, see \cite[Section 4]{Yessen2011a}. Notice that $\mathfrak{p} = 1$ corresponds to the Schr\"odinger case, and $\mathfrak{q} = 0$ to the off-diagonal case studied in \cite[Appendix A]{Damanik2010a} and \cite{Dahl2010}. Also notice the difference in the bound on the right from that in Theorems \ref{thm:schro-spect-haus-cont-v0} and \ref{thm:mei2}. Indeed, in the theorem above the bound on the right is sharp (see Theorem \ref{thm:mei-yessen-7} below).
\end{rem}

\begin{thm}\label{thm:mei-yessen-5}
 Let $\alpha(t) = (\mathfrak{p}(t), \mathfrak{q}(t)): [0, 1]\rightarrow \R^2$ be a regular curve, such that for all $t$, $\mathfrak{p}(t)\neq 0$, $\alpha(0) = (1,0)$, and for all $t\in (0, 1]$, $\alpha(t)\neq (1,0)$. Then there exists $\delta > 0$ such that if for all $t\in [0, 1]$, $\norm{\alpha(t) - \alpha(0)} \leq \delta$, then the following holds.

 Fix a gap of $\Sigma_{v, \alpha(1)}$ and call it $U_{\alpha(1)}$. Continue this gap along $\alpha$, obtaining a family of gaps $U_{\alpha(t)}$, $t\in (0, 1]$. Then the boundary points of $U_{\alpha(t)}$ as a function of $t$ are of the same smoothness class as $\alpha(t)$, and
 \begin{align*}
  \lim_{t\rightarrow 0}\frac{\abs{U_{\alpha(t)}}}{\norm{\alpha(t) - (1,0)}}\hspace{2mm}\text{ exists and belongs to}\hspace{2mm} (0, \infty),
 \end{align*}
 where $\norm{\cdot}$ is the standard norm on $\R^2$.
\end{thm}
\begin{rem}
The limit in the theorem above may depend on the gap, but does not depend on the choice of $\alpha(t)$.
\end{rem}

\begin{thm}\label{thm:mei-yessen-7}
 For any primitive invertible substitution sequence $v$ and all $(\mathfrak{p},\mathfrak{q})\neq (1,0)$, $\Sigma_{v, (\mathfrak{p},\mathfrak{q})}$ is a multifractal; more precisely, the following statements hold.
 \begin{enumerate}[(1)]
 \setlength{\itemsep}{2pt}

  \item $\lhdim(\Sigma_{v, (\mathfrak{p},\mathfrak{q})}, a)$, as a function of $a\in \Sigma_{v,(\mathfrak{p},\mathfrak{q})}$, is continuous; it is constant if $\mathfrak{p} = 1$ or $\mathfrak{q} = 1$, and non-constant otherwise.

  \item There exists a nonempty set $\mathfrak{N}\subset\R^2$ of Lebesgue measure zero, such that the following holds.
  \begin{enumerate}[(a)]
  \setlength{\itemsep}{2pt}

   \item For all $(\mathfrak{p},\mathfrak{q})\notin\mathfrak{N}$, we have $0 < \lhdim(\Sigma_{v, (\mathfrak{p},\mathfrak{q})}, a) < 1$ for all $a\in \Sigma_{v, (\mathfrak{p},\mathfrak{q})}$; hence we have $0 < \hdim(\Sigma_{v, (\mathfrak{p},\mathfrak{q})}) < 1$.

   \item For $(\mathfrak{p},\mathfrak{q})\in\mathfrak{N}$, $0 < \lhdim(\Sigma_{v, (\mathfrak{p},\mathfrak{q})}, a) < 1$ for all $a\in \Sigma_{v, (\mathfrak{p},\mathfrak{q})}$ away from the lower and upper boundary points of the spectrum, and $\hdim(\Sigma_{v, (\mathfrak{p},\mathfrak{q})}) = 1$. In fact, the dimension accumulates at one of the endpoints of the spectrum.
  \end{enumerate}
 \end{enumerate}
\end{thm}

\begin{rem}
Let us remark that in 2(b) of Theorem \ref{thm:mei-yessen-7}, it can be determined whether the dimension accumulates at the lower or the upper endpoint of the spectrum according to whether $\mathfrak{q}(\mathfrak{p}^2-1)$ is positive or negative, respectively. Indeed, this follows since the dimension must accumulate at the point $E$ in the spectrum where $I\circ l(E) = 0$, with $l$ being the curve of initial conditions from \eqref{eq:ci}. On the other hand, the function $I\circ l: \R\rightarrow\R$ is monotone, with the sign of the derivative being given by the sign of $\mathfrak{q}(\mathfrak{p}^2-1)$; see equation \eqref{eq:fv-de} below.
\end{rem}

For the following theorem, see Figure \ref{fig:regions} to aid with visualization. Also see Remark \ref{rem:regions} following Theorem \ref{thm:yessen2} above.

\begin{thm}\label{thm:mei-yessen-8}
 The following statements hold for every primitive invertible substitution sequence $v$.
 \begin{enumerate}[(1)]
 \setlength{\itemsep}{2pt}

  \item There exists $\epsilon > 0$ such that for all $(\mathfrak{p}, \mathfrak{q})$ satisfying $\norm{(1,0) - (\mathfrak{p},\mathfrak{q})} < \epsilon$, the box-counting dimension of $\Sigma_{v, (\mathfrak{p},\mathfrak{q})}$ exists and coincides with the Hausdorff dimension.

  \item There exists $\Delta > 0$, such that for all $\abs{\mathfrak{p}}\geq \Delta$, there exists $\delta_{\mathfrak{p}} > 0$, so that for all $\abs{\mathfrak{q}} < \delta_{\mathfrak{p}}$, the box-counting dimension of $\Sigma_{v, (\mathfrak{p},\mathfrak{q})}$ exists and coincides with the Hausdorff dimension.

  \item There exists $\Delta > 0$, such that for all $\abs{\mathfrak{q}}\geq \Delta$, there exists $\delta_{\mathfrak{q}} > 0$, so that for all $\abs{\mathfrak{p}-1} < \delta_{\mathfrak{q}}$, the box-counting dimension of $\Sigma_{v, (\mathfrak{p},\mathfrak{q})}$ exists and coincides with the Hausdorff dimension.
 \end{enumerate}
\end{thm}

\begin{thm}\label{thm:mei-yessen-9}
 For every primitive invertible substitution sequence $v$ and $(\mathfrak{p},\mathfrak{q})\in\R^2$, there exists $\mathfrak{D}_{v, (\mathfrak{p},\mathfrak{q})}\subset\R$ of full $dN$ measure, such that for all $E\in \mathfrak{D}_{v, (\mathfrak{p},\mathfrak{q})}$, we have
 \begin{align*}
  \lim_{\epsilon\downarrow 0}\frac{\log N(E - \epsilon, E+\epsilon)}{\log\epsilon} = d_{v, (\mathfrak{p},\mathfrak{q})}(E)\in\R,
 \end{align*}
 $d_{v, (\mathfrak{p},\mathfrak{q})}(E) > 0$. Moreover, if $(\mathfrak{p},\mathfrak{q})\neq(1,0)$, then
 \begin{align*}
  d_{v,(\mathfrak{p},\mathfrak{q})}(E) < \lhdim(\Sigma_{v,(\mathfrak{p},\mathfrak{q})}, E).
 \end{align*}
 Also,
 \begin{align*}
  \lim_{(\mathfrak{p},\mathfrak{q})\rightarrow(1,0)}\sup_{E\in\mathfrak{D}_{v,(\mathfrak{p},\mathfrak{q})}}\set{d_{v,(\mathfrak{p},\mathfrak{q})}(E)} = \lim_{(\mathfrak{p},\mathfrak{q})\rightarrow(1,0)}\inf_{E\in\mathfrak{D}_{v,(\mathfrak{p},\mathfrak{q})}}\set{d_{v,(\mathfrak{p},\mathfrak{q})}(E)} = 1.
 \end{align*}
\end{thm}

Notice that so far we have been working under the assumption that $\mathfrak{p}\neq 0$. It is interesting to see what happens in the case $\mathfrak{p}=0$, and how such quantities as the Hausdorff dimension change in the limiting process $\mathfrak{p}\rightarrow 0$. In fact, when $\mathfrak{p}=0$ it is easy to see that the Hamiltonian decouples into an infinite direct sum of finite-dimensional operators, and the spectrum consists of finitely many points (eigenvalues of infinite multiplicity). In this case the Hausdorff dimension of the spectrum is zero. We have the following theorem, which we prove in the next section.

\begin{thm}\label{thm:p-zero}
For any fixed $\mathfrak{q}$ and any substitution sequence $v$, we have
\begin{align*}
\lim_{\mathfrak{p}\rightarrow 0}\hdim(\Sigma_{v, (\mathfrak{p},\mathfrak{q})})=0.
\end{align*}
\end{thm}

Thus we have continuity of the Hausdorff dimension. An interesting question is whether we also have continuity in the Hausdorff metric (that is, whether as $\mathfrak{p}\rightarrow 0$ the spectrum continuously collapses onto a set of finitely many points). The answer is \textit{yes} and follows from simple spectral-theoretic arguments (indeed, $H_{v, (\mathfrak{p}, \mathfrak{q})}$ converges to $H_{v, (0, \mathfrak{q})}$ in norm as $\mathfrak{p}\rightarrow 0$). Let us record this result as

\begin{thm}\label{thm:p-zero2}
For any substitution sequence $v$ and any fixed $\mathfrak{q}$, the spectrum of $H_{v, (\mathfrak{p},\mathfrak{q})}$ converges in Hausdorff metric to the spectrum of $H_{v, (0, \mathfrak{q})}$.
\end{thm}

\section{Proofs of new results.}\label{sec:proofs}

Before we begin proving the results from the previous section, we need to set up the machinery of trace maps, as our proofs will rely in an essential way on the dynamical properties of these maps. Moreover, a certain amount of preparation in hyperbolic and partially hyperbolic dynamics, including notation and terminology, is needed. Rather than satisfying these prerequisites here by adding an extra section, we chose to maintain the focus on the task at hand; instead, we invite the reader to follow along \cite[Appendix B]{Yessen2011}. Below we shall use the results from (partially) hyperbolic dynamics freely, pointing the reader to the appropriate sections, definitions and theorems in the aforementioned paper, as necessary (the interested readers will find a comprehensive overview of the theory, together with the proofs of the main results, in, for example, \cite{Katok1995,Hasselblatt2002,Hasselblatt2002b,Hasselblatt2006,Pesin2004} and, for the dimension theory in dynamical systems, \cite{Pesin1997}).

\subsection{Trace map dynamics}\label{sec:t-map-dynamics}

In what follows, it is convenient to view a given primitive and invertible substitution $s$ (acting on two letters) as a morphism on the free group $\Gamma \eqdef \langle 0, 1 \rangle$. Also, let us assume without loss of generality that the word $s(0)$ begins with $0$ (otherwise, we can assume that either $s(1)$ begins with $1$ or replace $s$ with $s^2$). As has been mentioned above, to each such $s$ there corresponds a map $T_s\in \mathbb{Z}[x,y,z]$ with the following properties.

\it
\begin{itemize}
\setlength{\itemsep}{2pt}

\item Given a representation $\rho: \Gamma \rightarrow \SL(2, \C)$, the trace of $\rho(s^n(0))$, $\Tr\rho(s^n(0))$, is given by
\begin{align*}
\frac{1}{2}\Tr\rho(s^n(0)) = \pi_1\circ T_s^n\left(\frac{1}{2}\Tr \rho (s(0)), \frac{1}{2}\Tr \rho (s(1)), \frac{1}{2}\Tr \rho (s(01))\right),
\end{align*}
where $\pi_1$ denotes the projection onto the first coordinate.

\rm

In this case $T_s^n$ is not quite iteration of $T_s$. As we shall see later, $T_s$ is composed of two polynomial maps, let us call them $T_a$ and $T_b$: $T_s = T_a\circ T_b$. In this case, $T_s^n$ denotes $T_a^n\circ T_b$, where $T_a^n$ is the $n$-fold composition of $T_a$ with itself.

\it

\item $T_s$ is invertible.

\item The map $T_s$ (and, in fact, its components $T_a$ and $T_b$, as above) preserves the Fricke-Vogt invariant $I$ from equation \eqref{eq:fv-invariant}, and hence also the surfaces $S_V$ from \eqref{eq:inv-surfaces}.

\end{itemize}
\rm

Since we shall be concerned with dynamics of the trace maps $T_s$ when restricted to the invariant surfaces, it is convenient to introduce the following notation, which we fix for the remainder of the paper.
\begin{align*}
T_{s, V} \eqdef T_s|_{S_V}\hspace{2mm}\text{ and }\hspace{2mm} \mathbb{S}_0 \eqdef S_0\cap [-1, 1]^3.
\end{align*}
In much of the theory, the map $T_{s,V}$, with $V$ positive but sufficiently close to zero, is investigated as a perturbation of the map $T_{s, 0}$. It is well know that the map $T_{s, 0}$ has the following properties.
\it
\begin{itemize}
\setlength{\itemsep}{2pt}

\item $T_{s, 0}$ fixes the Morse-type singularity $P_1=(1,1,1)$ of $S_0$ and permutes the other three (namely, $P_2=(1, -1, -1), P_3=(-1, 1, -1)$ and $P_3=(-1, -1, 1)$) in a three-cycle.

\item Restricted to the ``inner'' part of the Cayley cubic $S_0$ (that is, the topological sphere with the four Morse-type singularities, $\mathbb{S}_0$), $T_{s, 0}$ is of pseudo-Anosov type (i.e. there exist two transverse measured foliations, invariant under $T_{s,0}$, one contracting and the other expanding under $T_{s, 0}$, called the \rm stable \it and the \rm unstable \it foliations, respectively; see the seminal works \cite{Thurston1988, Penner1988}, and for textbook expositions see, for example, \cite{Fathi1979} and \cite[Part 3]{Farb2012}).

\item The map $T_{s, 0}$ is a factor of the Anosov automorphism $\mathcal{A}_s$ on the two-dimensional torus, $\mathbb{T}^2$; that is, the following diagram commutes,
\begin{center}
\begin{tikzpicture}
\matrix(m)[matrix of math nodes,row sep=3em,column sep=4em,minimum width=2em]
{
	\mathbb{T}^2 & \mathbb{T}^2\\
	\mathbb{S}_0 & \mathbb{S}_0\\
};
\path[-stealth]
	(m-1-1) edge node [left] {$F$} (m-2-1)
	(m-1-1) edge node [above] {$\mathcal{A}_s$} (m-1-2)
	(m-1-2) edge node [right] {$F$} (m-2-2)
	(m-2-1) edge node [below] {$T_{s, 0}$} (m-2-2);
\end{tikzpicture}
\end{center}
where
\begin{align}\label{eq:factor}
F(\theta, \phi) = (\cos2\pi(\theta + \phi),\cos2\pi\theta, \cos2\pi\phi),
\end{align}
and
\begin{align*}
\mathcal{A}_s =
\begin{pmatrix}
\#_a s(0) & \#_b s(0)\\
\#_a s(1) & \#_b s(1)
\end{pmatrix}
\end{align*}
with $\#_{\star_1} s(\star_2)$ denoting the number of letters $\star_1$ in the word $s(\star_2)$.

\item There are four smooth (in fact, analytic) one-dimensional manifolds $W_{i = 1, \dots, 4}^s$ injectively immersed in $\mathbb{S}_0\setminus\set{P_i},$ where $i = 1, \dots, 4$, such that for every $x$ in $W_i^s$, $T^n_{s, 0}(x)$ converges to $P_i$ exponentially fast as $n\rightarrow\infty$. These manifolds, $W_i^s$, are each dense in the stable pseudo-Anosov foliation of $\mathbb{S}_0$. There also exist four manifolds $W_i^u$ which exhibit the same properties as $W_i^s$ with $T_{s, 0}$ replaced by $T_{s, 0}^{-1}$.

\item The periodic points of $T_{s,0}$ form a dense subset of $\mathbb{S}_0$.

\item As a consequence of $T_{s, 0}$ being pseudo-Anosov, every periodic point that is not one of the singularity points is of saddle type with one contracting and one expanding tangent directions (i.e. the eigenvalues of the differential do not lie on the unit circle). Consequently, each such point can be smoothly (in fact, analytically) continued to the surfaces $S_V$ for $V\in (-\delta, \delta)$, $\delta > 0$, but $\delta$ depends on the point.
\end{itemize}
\rm
\begin{rem}
Notice that the map $F$ above is not a conjugacy in the sense of \cite[Appendix B.1.1]{Yessen2011}, but a \textit{semi-conjugacy}, since $F$ is not invertible; in fact, $F$ is a branched double cover of the sphere, with four ramification points. At those ramification points (the four singularities), $F$ is certainly not smooth; nontheless, $F$ is ``sufficiently conformal'' near the singularities to push the stable and unstable foliations of $\mathbb{T}^2$ from $\mathbb{T}^2$ to $\mathbb{S}_0$ to produce pseudo-Anosov behavior on $\mathbb{S}_0$. For the formal statement, see \cite[Lemma 3.1]{Damanik2009}).
\end{rem}

It turns out that the dynamics on the surfaces $S_V$ for $V > 0$ is more complicated. Let us describe (some of) the known results that we shall use shortly.
\it
\begin{itemize}
\setlength{\itemsep}{2pt}
\item $T_{s, V}$, for all $V > 0$, satisfies S. Smale's Axiom A (see Smale's original work \cite{Smale1967} for definitions, properties and examples). This result is due to M. Casdagli for all $V > 0$ sufficiently large \cite{Casdagli1986}, D. Damanik and A. Gorodetski for all $V > 0$ sufficiently small \cite{Damanik2009}, and S. Cantat for all $V > 0$ uniformly \cite{Cantat2009}.

\item The nonwandering set of $T_{s, V}$ coincides with the set of those points whose full orbit is bounded (see \cite[Section I.1]{Smale1967} for terminology; this result is implicit in \cite{Casdagli1986} and proved explicitly in \cite[Section 2]{Damanik2013a} for the case of $s$ being the Fibonacci substitution. With little effort, the result can be verified for other primitive invertible two-letter substitutions--see also \cite{Cantat2009}).

\item As a consequence of the above two points, there exists a one-dimensional lamination in $S_V$, with smooth (in fact, analytic) leaves, such that a point of $S_V$ has a bounded forward orbit if and only if it belongs to a leaf of this lamination. This lamination is called the \textrm{stable} lamination. There also exists the analogous unstable lamination with the same properties, with \textrm {forward orbit} replaced by \textrm{backward orbit}. We shall denote these laminations by
\begin{align*}
\mathcal{L}_V^s \hspace{2mm} \text{ and } \hspace{2mm} \mathcal{L}_V^u,
\end{align*}
respectively.

\item The laminations $\mathcal{L}_V^{s,u}$ are invariant under $T_V^{\pm}$. This in particular implies that every point in $S_V$ either has a bounded forward orbit, or escapes to infinity in every one of its three coordinates (i.e. an unbounded orbit does not contain bounded subsequences) \cite{Roberts1996}.

\item The laminations $\mathcal{L}_V^{s,u}$, when viewed as sets, are closed.

\item Any smooth, regular, compact curve of $S_V$ that intersects the lamination $\mathcal{L}_V^\bullet$, $\bullet\in \set{s, u}$, transversally in its interior, intersects it in a dynamically defined Cantor set along its interior (see the paragraph following Definition \ref{defn:type-b} above for terminology).

\item The holonomy map along the lamination $\mathcal{L}_V^\bullet$, $\bullet\in \set{s, u}$ between any two smooth curves that intersect $\mathcal{L}_V^\bullet$ uniformly transversally is Lipschitz continuous (in fact, it can be extended to a $C^1$ map) as a consequence of \cite[Theorem 6.4]{Hirsch1968}.
\end{itemize}
\rm
Notice that from above we get complete characterization of the forward-stable set on each $S_V$, $V > 0$ (i.e. the set of those points that do not escape to infinity in forward time; recall that we called these points type-\textbf{B} above). We shall find that while investigating the spectrum of the operator of type \eqref{eq:qp-Jacobi}, rather than \eqref{eq:schrodinger}, we must consider the action of $T$ on $S_V$ collectively for all $V > 0$ (i.e. the action of $T$ on the smooth three-manifold $\bigcup_{V > 0}S_V$). In particular, we shall need a (geometric) classification of stable points in $\bigcup_{V > 0}S_V$. Before we continue, let us set and fix for the remainder of the paper the following notation.
\begin{align*}
\mathcal{M} \eqdef \bigcup_{V > 0}S_V.
\end{align*}
We have the following (for details, refer to \cite{Yessen2011}).
\it
\begin{itemize}
\setlength{\itemsep}{2pt}
\item $T$ restricted to $\mathcal{M}$ is partially hyperbolic on the nonwandering set of $\mathcal{M}$ (\cite{Smale1967} for terminology). Given any leaf $L_{V_0}$ of the lamination $\mathcal{L}_{V_0}^{s}$ (respectively, $\mathcal{L}_{V_0}^u$) passing through a point $x_{V_0}$ in the set of bounded orbits on $S_{V_0}$, $\bigcup_{V > 0}L_V$ is a smooth (in fact analytic) injectively immersed two-dimensional submanifold of $\mathcal{M}$; here $L_V$ denotes the leaf of the stable (respectively, unstable) lamination passing through $h_V(x_{V_0})\in S_V$, where $h_V$ is the homeomorphism forming the topological conjugacy from the set of bounded orbits on $S_{V_0}$ to that on $S_V$ (see \cite[Appendix B]{Yessen2012a} for details). The resulting two-dimensional manifold is called the \textnormal{center-stable} (respectively, \textnormal{center-unstable}) manifold.

\item The union of all the center-stable manifolds (and the center-unstable ones) forms a closed set in $\mathcal{M}$.

\item The forward-stable (respectively, backward-stable) points of $\mathcal{M}$ are precisely those that lie on the center-stable (respectively, center unstable) manifolds.

\item Any smooth, compact, regular curve in $\mathcal{M}$ that intersects the center-stable (or center-unstable) manifolds transversally in its interior, intersects these manifolds in a Cantor set (in general, this Cantor set is not dynamically defined -- the local Hausdorff dimension may depend on the point of localization; the box-counting dimension, however, still exists and coincides with the Hausdorff).

\end{itemize}
\rm
Let us now describe the formalism that allows one to pass from the spectral-analytic problem to a problem in dynamical systems. We have already described this for the Schr\"{o}dinger operator in Section \ref{sec:known}. Let us do the same for the Jacobi operators.

\subsection{The trace map formalism for the substitution Jacobi operators}
For a given substitution sequence $v$ generated by a primitive invertible two-letter substitution $s$ (again, without loss of generality we assume that $s(0)$ starts with the letter $0$), and $(\mathfrak{p}, \mathfrak{q})$ as above, let $H_{v, (\mathfrak{p},\mathfrak{q})}^{(k)}$ denote the $k$th periodic approximation of $H_{v, (\mathfrak{p},\mathfrak{q})}$ and denote its spectrum by $\sigma_k$ as in Theorem \ref{thm:mei-yessen-2a}.

Now let $\psi(E)$ and $\phi(E)$ be two solutions of the eigenvalue equation
\begin{align*}
H_{v, (\mathfrak{p},\mathfrak{q})}^{(k)}\theta = E\theta
\end{align*}
with $\phi_0(E) = \psi_{-1}(E) = 1$ and $\phi_{-1}(E) = \psi_0(E) = 1$. By the Floquet theory \cite{Toda1981,Teschl1999}, we have
\begin{align}\label{eq:eigenvalue}
\sigma_k = \set{E\in\R: \frac{1}{2}\abs{\phi_{\abs{s^k(0)}}(E) + \psi_{\abs{s^k(0)}}} \leq 1}
\end{align}
(recall that $H_{v, (\mathfrak{p},\mathfrak{q})}^{(k)}$ is periodic of period $\abs{s^k(0)}$).

With $p, q: \set{0,1}\rightarrow\R$ as defined just prior to the statement of Theorem \ref{thm:fib-jacobi-sc-spectr} and $v^{(k)}$ as in Theorem \ref{thm:mei-yessen-2a}, let us write $p_{k,n}$ for $p(v^{(k)}_n)$, and similarly for $q$. Define the two matrices
\begin{align}\label{eq:transfer-mat}
\begin{split}
M^{(n)}_{v,(\mathfrak{p},\mathfrak{q})}(E) &\eqdef \frac{1}{p_{k,n+1}}
\begin{pmatrix}
	E - q_{k,n}		&	-p_{k, n}\\
	p_{k,n+1}			&	0\\
\end{pmatrix}\\[4pt]
\hspace{2mm}
T^{(n)}_{v,(\mathfrak{p},\mathfrak{q})}(E) &\eqdef \frac{1}{p_{k,n+1}}
\begin{pmatrix}
	E - q_{k,n}		&	-1\\
	p_{k,n+1}^2 		&	0\\
\end{pmatrix}
\end{split}
\end{align}
and let $\Theta_n = (\theta_n, p_{k,n}\theta_{n-1})^{\mathrm{T}}$. A simple computation shows that $\theta$ satisfies the equation \eqref{eq:eigenvalue} if and only if
\begin{align}\label{eq:eigen-transfer}
\begin{pmatrix}
	\theta_n\\
	\theta_{n-1}\\
\end{pmatrix}
= M^{(n)}_{v,(\mathfrak{p},\mathfrak{q})}
\begin{pmatrix}
	\theta_{n-1}\\
	\theta_{n-2}\\
\end{pmatrix}
\iff
\Theta_n = T^{(n)}_{v,(\mathfrak{p},\mathfrak{q})}\Theta_{n-1}\hspace{2mm}\text{ for all }\hspace{2mm}n\in\Z.
\end{align}
Let us define
\begin{align}\label{eq:long-transfer}
\hat{T}^{(k)}_{v, (\mathfrak{p},\mathfrak{q})}(E) \eqdef T^{(\abs{s(0)})}_{v,(\mathfrak{p},\mathfrak{q})}(E)\cdots T^{(1)}_{v,(\mathfrak{p},\mathfrak{q})}(E).
\end{align}
From \eqref{eq:eigen-transfer} we have $\Theta_{\abs{s(0)}} = \hat{T}^{(k)}_{v, (\mathfrak{p},\mathfrak{q})}\Theta_0$; hence using $\phi$ and $\psi$ in place of $\theta$ we get $\phi_{\abs{s(0)}}= [\hat{T}^{(k)}_{v, (\mathfrak{p},\mathfrak{q})}]_{11}$ and $p_{k, \abs{s(0)}}\psi_{\abs{s(0)-1}}= p_{k, 0}[\hat{T}^{(k)}_{v, (\mathfrak{p},\mathfrak{q})}]_{22}$. Since $v^{(k)}$ is $\abs{s(0)}$-periodic, $p_{k, \abs{s(0)}} = p_{k, 0}$, so we have
\begin{align*}
\frac{1}{2}\abs{\phi_{\abs{s(0)}} + \psi_{\abs{s(0)}-1}} = \frac{1}{2}\abs{\Tr\hat{T}^{(k)}_{v, (\mathfrak{p},\mathfrak{q})}(E)}.
\end{align*}
Thus the spectrum of the $k$th periodic approximation is given by
\begin{align}\label{eq:k-spect}
\sigma_k = \set{E\in\R: \frac{1}{2}\abs{\Tr\hat{T}^{(k)}_{v, (\mathfrak{p},\mathfrak{q})}(E)}\leq 1}.
\end{align}
Now let us define the following set, to which in the remainder of the paper we shall refer as \textit{the dynamical spectrum}:
\begin{align}\label{eq:dyn-spect}
B_{v, (\mathfrak{p},\mathfrak{q})} \eqdef\set{E\in\R: \abs{\hat{T}^{(k)}_{v, (\mathfrak{p},\mathfrak{q})}(E)}\hspace{2mm}\text{is bounded}}.
\end{align}
One of the main ingredients in our proofs will be the following lemma.
\begin{lem}\label{lem:main}
Given a substitution sequence $v$ generated by a primitive invertible two-letter substitution $s$, we have
\begin{align}\label{eq:dyn-real-spect-equality}
\sigma(H_{v, (\mathfrak{p},\mathfrak{q})}) = B_{v,(\mathfrak{p},\mathfrak{q})}.
\end{align}
\end{lem}
Let us postpone the proof of this lemma until the next section.

Notice that the matrix ${T}_{v, (\mathfrak{p},\mathfrak{q})}^{(k)}(E)$ is an element of $\SL(2, \R)$. This gives us a representation $\rho:\Gamma\rightarrow \SL(2,\R)$, with
\begin{align*}
\rho(0) = \begin{pmatrix}
E	&	-1\\
1	&	0\\
\end{pmatrix}
\hspace{2mm}\text{ and }\hspace{2mm}
\rho(1) = \frac{1}{\mathfrak{p}}\begin{pmatrix}
E - \mathfrak{q}	&	-1\\
\mathfrak{p}^2		&	0
\end{pmatrix}.
\end{align*}
Let us set
\begin{align*}
x_1= \frac{1}{2}\Tr\rho(0);\hspace{4mm}
x_2= \frac{1}{2}\Tr\rho(1);\hspace{4mm}
x_3= \frac{1}{2}\Tr\rho(01).
\end{align*}
Let us also define \textit{the curve of initial conditions} $l_{v, (\mathfrak{p},\mathfrak{q})}:\R\rightarrow\R^3$ by
\begin{align}\label{eq:line}
l_{v, (\mathfrak{p},\mathfrak{q})}(E) = (x_3(E), x_2(E), x_1(E)).
\end{align}
Now if $T_s$ is the trace map that is associated to $s$, from Lemma \ref{lem:main} we get the following characterization of the spectrum.
\begin{align}\label{eq:dyn-to-real}
\sigma(H_{v,(\mathfrak{p},\mathfrak{q})}) = \set{E\in\R: \mathcal{O}_{T_s}^+(l_{v,(\mathfrak{p},\mathfrak{q})}(E))\hspace{2mm} \text{is bounded}}.
\end{align}
The explicit expression for the curve of initial conditions is
\begin{align}\label{eq:ci}
l_{v,(\mathfrak{p},\mathfrak{q})} = \left(\frac{E^2 - \mathfrak{q}E - \mathfrak{p}^2 - 1}{2\mathfrak{p}}, \frac{E - \mathfrak{q}}{2\mathfrak{p}}, \frac{E}{2} \right)
\end{align}
Computation of the Fricke-Vogt invariant $I$ along the curve $l_{v, (\mathfrak{p},\mathfrak{q})}$ gives
\begin{align}\label{eq:Ici}
I\circ l_{v,(\mathfrak{p},\mathfrak{q})}(E) = \frac{\mathfrak{q}(\mathfrak{p}^2 - 1)E + \mathfrak{q}^2 + (\mathfrak{p}^2 - 1)^2}{4\mathfrak{p}^2}.
\end{align}
Notice that the value of $I$ depends on $E$ in general, and is independent of $E$ if either $\mathfrak{q} = 0$ (the so-called off-diagonal case that has been mentioned in the introduction, and has been studied in \cite{Damanik2010a} and in \cite{Dahl2010}), or if $\mathfrak{p} = 1$ (the Schr\"odinger case that has been described in Section \ref{sec:known} above in the case of the Fibonacci substitution $s$, and has been extended (at least in the small coupling regime, i.e. for $\mathfrak{q}$ sufficiently close to zero) to general two-letter substitutions in \cite{Mei2013}). The case of $s$ being the Fibonacci substitution with $\mathfrak{q}\neq 0$ \textit{and} $\mathfrak{p}\neq 1$ has been investigated in \cite{Yessen2011a}.
\it
\begin{itemize}
\setlength{\itemsep}{2pt}
\item From now on we assume that $\mathfrak{p}\neq 1$ and $\mathfrak{q}\neq 0$.
\end{itemize}
\rm

Notice that $l_{v,(\mathfrak{p},\mathfrak{q})}$ is completely independent of $v$. Henceforth we drop $v$ from the notation and write simply $l_{(\mathfrak{p},\mathfrak{q})}$. If we differentiate $I\circ l_{(\mathfrak{p},\mathfrak{q})}$, we get
\begin{align}\label{eq:fv-de}
\frac{\partial (I\circ l_{(\mathfrak{p},\mathfrak{q})})}{\partial E} = \frac{\mathfrak{q}(\mathfrak{p}^2 - 1)}{4\mathfrak{p}^2}.
\end{align}
This quantity is clearly nonzero and it shows that $I\circ l_{(\mathfrak{p},\mathfrak{q})}$ is either increasing or decreasing with $E$, depending on whether $\mathfrak{q}(\mathfrak{p}^2 - 1)$ is positive or negative. Moreover, since the expression $I(l_{(\mathfrak{p},\mathfrak{q})})$ is linear in $E$, the equation $I(l_{(\mathfrak{p},\mathfrak{q})}) = V$ can easily be solved for $E$, with any $V\in\R$. This shows that the line $l_{(\mathfrak{p},\mathfrak{q})}$ passes through every surface $S_V$, $V\in\R$, and intersects each surface precisely once. This is exactly the complication that prevents verbatim generalization of the previous techniques (recall that in the Schr\"odinger case, as well as the off-diagonal case, the corresponding curve of initial conditions always lies on the single surface).

From the facts stated in the previous section together with Lemma \ref{lem:main}, it is evident that the spectrum $\sigma(H_{v, (\mathfrak{p},\mathfrak{q})})$ consists of those values $E\in\R$ for which $\mathcal{O}_{T_s}^+(l_{(\mathfrak{p},\mathfrak{q})})$ is bounded. We recall that in $\mathcal{M}$, we have complete characterization of such orbits: these are precisely the center-stable manifolds. However, $l_{(\mathfrak{p},\mathfrak{q})}$ intersects $\mathcal{M}^c$ (i.e. the surfaces $S_V$ with $V < 0$). What can be said about bounded orbits there?

In this case, the corresponding surface $S_V$ consists of five smooth connected components: one a topological sphere (which degenerates to a point at the origin at $V = -1$ and disappears for $V < -1$) contained in the interior of the unit cube centered at the origin, and the other four are topological discs. Since the associated trace map $T_s$ is invertible (and hence a homeomorphism) and leaves the surfaces $S_V$ invariant, it is easy to see that the spherical component is invariant. On the other hand, it is known that every point on the discs escapes to infinity \cite{Roberts1996}. Since we shall only be concerned with type-\textbf{B} points, we need to concentrate only on the spherical component. On the other hand, should $l_{(\mathfrak{p},\mathfrak{q})}$ intersect the spherical component of $S_{V_0}$ for some $V_0 < 0$, then $l_{(\mathfrak{p},\mathfrak{q})}$ must also intersect the spherical components of $S_V$ for $V\in(V_0 - \delta, V_0 + \delta)$ for $\delta > 0$ sufficiently small. This would produce an interval in the spectrum of $H_{v, (\mathfrak{p},\mathfrak{q})}$. However, as we shall soon see, the spectrum of $H_{v, (\mathfrak{p},\mathfrak{q})}$ is a Cantor set. In summary, the part of $l_{(\mathfrak{p},\mathfrak{q})}$ that lies in $\bigcup_{V < 0}S_V$ does not contain any type-\textbf{B} points.

The discussion above can be summarized as follows.
\it
\begin{itemize}
\setlength{\itemsep}{2pt}
\item Given a substitution sequence $v$ generated by a primitive invertible two-letter substitution $s$, the spectrum of $H_{v, (\mathfrak{p},\mathfrak{q})}$ consists of those energies $E\in\R$ for which $l_{(\mathfrak{p},\mathfrak{q})}(E)$ lies on a center-stable manifold inside $\mathcal{M}$ (actually, we shall see that there may be one $E\in\sigma(H_{v, (\mathfrak{p},\mathfrak{q})})$ lying in $S_0$, and in this case $E$ marks one of the extrema of the spectrum).

\end{itemize}
\rm

\subsection{Proofs}
With the tools from the previous section, we are now ready to prove the results of Section \ref{sec:main}

\subsubsection{Proof of Lemma \ref{lem:main}}\label{sec:lem:main}
This lemma generalizes Theorem 2.1 in \cite{Damanik2000}, we also refer the reader to \cite{Roberts1996}. By \cite{Hof1993} the Lyapunov exponent
$$\lim_{k \to \infty} \frac{1}{k} \log\norm{\hat{T}^{(k)}_{v, (\mathfrak{p},\mathfrak{q})}(E)}$$
exists for all energies $E$ (in fact uniformly for all $\omega\in\Omega_v$). Let
\begin{align}
A_{v,(\mathfrak{p},\mathfrak{q})} \eqdef \set{E \in \R: \lim_{k \to \infty} \frac{1}{k} \log\norm{\hat{T}^{(k)}_{v, (\mathfrak{p},\mathfrak{q})}(E)}=0}.
\end{align}
By Kotani theory, $A_{v,(\mathfrak{p},\mathfrak{q})}$ has zero Lebesgue measure and is a subset of $\Sigma_{v, (\mathfrak{p},\mathfrak{q})}$. See \cite{Kotani1989}, see \cite{Remling2011} for the extension to the Jacobi case, and see also \cite{Yessen2011a} for a discussion in the particular case of the Fibonacci Jacobi operator. The bulk of the proof of Theorem 2.1 in \cite{Damanik2000} is model-independent, i.e. the proof concerns the trace map without regard to the particular operator being studied. Here, we will show that $\Sigma_{v, (\mathfrak{p},\mathfrak{q})} \subseteq B_{v,(\mathfrak{p},\mathfrak{q})}$ in Lemma \ref{lem:main-1} and that $B_{v,(\mathfrak{p},\mathfrak{q})}\subseteq A_{v,(\mathfrak{p},\mathfrak{q})}$ in Lemma \ref{lem:main-2}. To keep the notation synchronized with that from Theorem 2.1 of \cite{Damanik2000}, let us introduce the following sets.
\begin{align*}
Esc_{v, (\mathfrak{p},\mathfrak{q})}^{(k)}(0)&\eqdef\set{E\in\R: \abs{x_k(E)}>1}\\
Esc_{v, (\mathfrak{p},\mathfrak{q})}^{(k)}(1)&\eqdef\set{E\in\R: \abs{y_k(E)}>1}\\
Esc_{v, (\mathfrak{p},\mathfrak{q})}^{(k)}(01)&\eqdef\set{E\in\R: \abs{z_k(E)}>1},
\end{align*}
where the $x_k = \pi_1\circ f^k(x_0,y_0,z_0)$, $y_k=\pi_2\circ f^k(x_0,y_0,z_0)$, and $z_k = \pi_3\circ f^k(x_0,y_0,z_0)$ with $\pi_i$ denoting projection onto the $i$th coordinate, and $(x_0, y_0, z_0)$ is the initial point $l_{v,(\mathfrak{p},\mathfrak{q})}(E)$ from \eqref{eq:ci}.
\begin{lem}\label{lem:main-1}
Let $v$ be a primitive invertible substitution sequence  and let $(\mathfrak{p},\mathfrak{q})\neq (1,0)$.
\begin{align*}\bigcup_{k \in \N} \Int\left(\bigcap_{m \geq k} Esc_{v, (\mathfrak{p},\mathfrak{q})}^{(k)}(0)\right) &\cup \bigcup_{k \in \N} \Int\left(\bigcap_{m \geq k} Esc_{v, (\mathfrak{p},\mathfrak{q})}^{(k)}(1)\right) \\
&\cup \bigcup_{k \in \N} \Int\left(\bigcap_{m \geq k} Esc_{v, (\mathfrak{p},\mathfrak{q})}^{(k)}(01)\right) \subseteq (\Sigma_{v, (\mathfrak{p},\mathfrak{q})})^c.
\end{align*}
\end{lem}
\begin{proof}
This is the direct analogue of Lemma 4.2 from \cite{Damanik2000}. We have previously introduced the notation $H_{v, (\mathfrak{p},\mathfrak{q})}^{(k)}$ for the $k$th periodic approximation of $H_{v, (\mathfrak{p},\mathfrak{q})}$. Let us introduce the notation $H_{\star, (\mathfrak{p},\mathfrak{q})}^{(k)}$ for the periodic approximation, using the sequence whose period is $\abs{s^k(\star)}$, constructed by gluing the words $s^k(\star)$ into an infinite sequence, in place of $v$; here $\star\in\set{0, 1, 01}$. For example, $H_{0, (\mathfrak{p},\mathfrak{q})}^{(k)}$ is exactly $H_{v, (\mathfrak{p},\mathfrak{q})}^{(k)}$. For each $n \in \N$, by the primitivity of $s$, there is some $n_k>0$ so that $v_{-n}...v_{n}$ is a word of $s^{n_k}(0)$, $s^{n_k}(1)$, and $s^{n_k}(01)$. Observe that $n_k$ can be chosen to be increasing so that as $k \to \infty$, $H_{x, (\mathfrak{p},\mathfrak{q})}^{(n_k)}$ converges strongly to $H_{v, (\mathfrak{p},\mathfrak{q})}$. Thus we have
\begin{align*}
\bigcup_{l \in \N} \Int \left(\bigcap_{k \geq l} \rho \left(H_{x, (\mathfrak{p},\mathfrak{q})}^{(n_k)}\right)\right) \subseteq (\Sigma_{v, (\mathfrak{p},\mathfrak{q})})^c
\end{align*}
and $\rho\left(H_{x, (\mathfrak{p},\mathfrak{q})}^{(n_k)}\right)$ is exactly $Esc_{v, (\mathfrak{p},\mathfrak{q})}^{(n_k)}(x)$ (here $\rho(\cdot)$ denotes the resolvent set of $\cdot$). This gives us our desired inclusion.
\end{proof}

\begin{lem}\label{lem:main-2}
Let $v$ be a primitive invertible substitution sequence  and let $(\mathfrak{p},\mathfrak{q})\neq (1,0)$. Then $B_{v,(\mathfrak{p},\mathfrak{q})}\subseteq A_{v,(\mathfrak{p},\mathfrak{q})}$
\end{lem}

\begin{proof}
This is the direct analogue of Lemma 4.23 from \cite{Damanik2000}. As the proof relies on Osceledec's theorem and Floquet theory-based considerations, the exact proof with $\hat{T}^{(k)}_{v, (\mathfrak{p},\mathfrak{q})}(E)$ in place of $M_E(S^k(0))$ (in the notation of \cite{Damanik2000}) gives us our desired result.
\end{proof}

\begin{proof}[Proof of Lemma \ref{lem:main}]
Lemmas \ref{lem:main-1} and \ref{lem:main-2}, together with Kotani theory from our comments at the beginning of this section, give us the following set inclusions
$$\Sigma_{v, (\mathfrak{p},\mathfrak{q})} \subseteq B_{v,(\mathfrak{p},\mathfrak{q})}\subseteq A_{v,(\mathfrak{p},\mathfrak{q})} \subseteq \Sigma_{v, (\mathfrak{p},\mathfrak{q})}$$
from which we conclude that $\Sigma_{v, (\mathfrak{p},\mathfrak{q})} = B_{v,(\mathfrak{p},\mathfrak{q})}$.
\end{proof}

\subsubsection{Proof of Theorem \ref{thm:mei-yessen-1}}
Let us begin by recalling some properties of substitution sequences. Let $w$ be a word of $v$ of length $n$ and let $... < m_{-1} < 0 \leq m_0 < m_1<...$ be the indices corresponding to occurrences of $w$, i.e. $v_{m_k}...v_{m_k+n-1}=w$.  We say that $v$ is \textit{recurrent} if $m_k \to \pm \infty$ as $k \to \pm\infty$ for every word $w$ of $v$. We say that $v$ is \textit{linearly recurrent} if there is a constant $C<\infty$ such that for every word $w$ of $v$, the gaps $m_{k+1}-{m_k}$ are bounded by $C|w|$. It is well known that if $v$ is a primitive invertible substitution sequence, then $v$ is linearly recurrent and the system $(\Omega_v, T)$ is strictly ergodic, i.e. minimal and uniquely ergodic. Recall these definitions from Section \ref{sec:intro}. For a proof of this based on linear recurrence, see for example \cite{Damanik2005}.
\begin{lem}\label{lem:indep-omega}
For every primitive invertible substitution sequence $v$ and $(\mathfrak{p},\mathfrak{q})\neq (1,0)$, there exists $\Sigma_{v, (\mathfrak{p},\mathfrak{q})}\subset\R$ such that for every $\omega\in\Omega_v$, $\Sigma_{\omega, (\mathfrak{p},\mathfrak{q})} = \Sigma_{v, (\mathfrak{p},\mathfrak{q})}$.
\end{lem}

The result of Lemma \ref{lem:indep-omega} is well known in the Schr\"odinger case, see for example Theorem 3.2 in \cite{Damanik 2005}. We follow along the same lines here.

\begin{proof}
Let $v_1, v_2 \in \Omega_v$. By the minimality of $(T, \Omega_v)$, there exists a sequence $n_k$ such that $T^{n_k}(v_2) \to v_1$ as $k \to \infty$. Note that $H_{T^{n_k}(v_2), (\mathfrak{p}, \mathfrak{q})}$ is unitarily equivalent to $H_{v_2, (\mathfrak{p}, \mathfrak{q})}$ for every $k$, thus $\Sigma_{T^{n_k}(v_2), (\mathfrak{p},\mathfrak{q})}=\Sigma_{v_2, (\mathfrak{p},\mathfrak{q})}$. Since $H_{T^{n_k}(v_2), (\mathfrak{p}, \mathfrak{q})}$ converges strongly to $H_{v_1, (\mathfrak{p}, \mathfrak{q})}$, we have the following set inclusions
$$\Sigma_{v_1, (\mathfrak{p}, \mathfrak{q})} \subseteq \bigcap_{l \in \N} \overline{\bigcup_{k \geq l} \Sigma_{T^{n_k}(v_2), (\mathfrak{p},\mathfrak{q})}} = \Sigma_{v_2, (\mathfrak{p},\mathfrak{q})}.$$
Hence, for any $v_1, v_2 \in \Omega_v$, $\Sigma_{v_1, (\mathfrak{p},\mathfrak{q})} \subseteq \Sigma_{v_2, (\mathfrak{p},\mathfrak{q})}$. By symmetry, the inclusion holds the other way too, so that for every $\omega\in\Omega_v$, $\Sigma_{\omega, (\mathfrak{p},\mathfrak{q})} = \Sigma_{v, (\mathfrak{p},\mathfrak{q})}$.
\end{proof}

\begin{lem}\label{lem:sc-spectrum}
$\Sigma_{v, (\mathfrak{p},\mathfrak{q})}$ is purely singular continuous.
\end{lem}

\begin{proof}
The combinatorial considerations in Proposition 5.1 \cite{DamanikLenz2003b} and an adaptation of a Gordon-type argument to the Jacobi operator prove empty point spectrum. Empty a.c. spectrum follows from the zero Lebesgue measure of the spectrum (see the next Lemma).
\end{proof}

\begin{lem}\label{lem:zero-meas-cantor}
$\Sigma_{v, (\mathfrak{p},\mathfrak{q})}$ is a Cantor set of zero Lebesgue measure.
\end{lem}

\begin{proof}
The spectrum is clearly compact. As a result of the proofs in Section \ref{sec:lem:main}, $\Sigma_{v, (\mathfrak{p},\mathfrak{q})}=A_{v, (\mathfrak{p},\mathfrak{q})}$, which has measure zero. Hence the spectrum is nowhere dense. Finally, any isolated point of $\Sigma_{v, (\mathfrak{p},\mathfrak{q})}$ would necessarily be an eigenvalue of $H_{v, (\mathfrak{p},\mathfrak{q})}$, which is precluded by Lemma ~\ref{lem:sc-spectrum}.
\end{proof}


\subsubsection{Proof of Theorem \ref{thm:mei-yessen-2a}}

In proving this result we shall follow the proof of Theorem 2.19 in \cite{Damanik2013a}. Our case requires some elaboration, since we are working with general trace maps, while the analogous result in \cite{Damanik2013a} was proved for the Fibonacci trace map only.

 By Lemma \ref{lem:main}, the dynamical spectrum coincides with the operator spectrum. On the other hand, we also have $\sigma_k = B_{v, (\mathfrak{p},\mathfrak{q})}^{(k)}$. Thus it is enough to show that $B_{v, (\mathfrak{p},\mathfrak{q})}^{(k)}\rightarrow B_{v, (\mathfrak{p},\mathfrak{q})}$ in the Hausdorff metric as $k\rightarrow\infty$.

 Let $\epsilon > 0$ arbitrary. Let $C_1, \dots, C_m$ be open sets of radius not larger than $\epsilon$ covering the compact set $B_{v, (\mathfrak{p},\mathfrak{q})}$; let us also assume that each $C_i$ intersects $B_{v, (\mathfrak{p},\mathfrak{q})}$, and the sets $C_i$ are pairwise disjoint (which can be insured by the Cantor structure of $B_{v, (\mathfrak{p},\mathfrak{q})}$). It is enough to show that there exists $N\in\N$ (depending only on $\epsilon$) such that for all $n\geq N$, $\sigma_n\cap C_i^c = \emptyset$ and $\sigma_n\cap C_i\neq\emptyset$ for all $i\in\set{1,\dots,m}$.

 Observe that since every point on the curve of initial conditions that is not contained in the dynamical spectrum escapes to infinity in every coordinate, certainly every point in $l_{(\mathfrak{p},\mathfrak{q})}\bigcap\left(\cup_{i}C_i^c\right)$ escapes to infinity in every coordinate. In particular, using the fact that the set of type-\textbf{B} points is closed, we know that there exists $N_0\in\N$ such that for all $n\geq N_0$, $\sigma_n\cap C_i^c = \emptyset$ for every $i$ (see \cite[Section 4]{Roberts1996} for more details).

 From \cite[Corollary 4.1]{Roberts1996} we know that the region
 \begin{align*}
 \mathcal{R}\eqdef\set{(x,y,z)\in\R^3: \abs{x},\abs{y},\abs{z} > 1}
 \end{align*}
 does not contain any periodic points. The following is known from \cite{Mei2013} (in the case of the Fibonacci trace map, this had been established earlier; see, for example, \cite{Damanik2009} and references therein).
 \it
 \begin{itemize}
 \setlength{\itemsep}{2pt}
 \item The fixed singularity $(1,1,1)$ on $S_0$ bifurcates into a curve of periodic points of period two; each of these points is a hyperbolic saddle. For all $V$ sufficiently close to zero, this curve of periodic points intersects the surfaces $S_V$ transversally.

 \item The three-cycle consisting of the other three singularities bifurcates into a six cycle (i.e. each of the three singularities bifurcates into a curve of period-six periodic points). These curves of periodic points have the same properties as the curve of period-two periodic points from above.

 \item The stable manifold of each of the eight periodic points from above (two period-two ones and six period-six ones) forms a dense sublamination of the stable lamination on $S_V$ (again, with $V > 0$ sufficiently close to zero).
 \end{itemize}
 \rm
 On the other hand, the above results can be pushed to higher values of $V$ by repeating the argument in the proof of \cite[Proposition 4.10]{Yessen2011} (in particular, transversality of the curves of periodic points with the surfaces $S_V$ for \textit{all} $V > 0$ holds as long as the periodic points are hyperbolic saddles, and the latter follows from \cite{Cantat2009}).

 Now, we know that these curves of periodic points do not intersect the region $\mathcal{R}$. By the density of the sublamination formed by the stable manifolds of the eight periodic points, there must exist $N_1 \in \N$ such that for all $n\geq N_1$, $\sigma_n$ intersects each of $C_i$. For slightly more technical details, see the proof of \cite[Theorem 2.19]{Damanik2013a}.

\subsubsection{Proof of Theorem \ref{thm:mei-yessen-3}}

To begin, we need the following technical lemma.

\begin{lem}\label{lem:qinter}
The center-stable manifold that contains the curve of period-two periodic points passing through the singularity $(1,1,1)$ is tangent to the Cayley cubic $S_0$, and the tangency is quadratic.
\end{lem}
\begin{rem}
We remark that the result of Lemma \ref{lem:qinter} is a very important technical ingredient in a number of works (e.g. \cite{Yessen2011a,Yessen2012a,Yessen2011,Damanik2013a,Damanik2010a}). This result is somewhat evident and has been used implicitly in the past. We include it here for completeness and clarity of exposition.
\end{rem}
\begin{proof}[Proof of Lemma \ref{lem:qinter}]
Let us introduce the new coordinates by shifting the fixed point $(1,1,1)$ to the origin, and parameterize the curve of period-two periodic points passing through the origin in these new coordinates as a regular curve $\rho: J\rightarrow U$, with $U$ a neighborhood of the origin, and $J\subset \R$ an open interval of finite length with $0\in J$. Assume that $\rho(0) = (1,1,1)$. In particular, $\rho'(0)\neq 0$ (see \cite{Mei2013} for more details). For $\alpha\in\set{x,y,z}$, denote the $\alpha$ coordinate of $\rho$ by $\rho_\alpha$. Observe that
\begin{align}\label{eq:arb1}
\frac{d}{dt}I\circ\rho(t) = 2((\rho_x - \rho_y\rho_z)\rho'_x + (\rho_y - \rho_x\rho_z)\rho'_y + (\rho_z - \rho_y\rho_x)\rho'_z)(t),
\end{align}
so that $\frac{d}{dt}I\circ\rho(0) = 0$. On the other hand, if we denote $a = \rho_x'$, $b = \rho_y'$ and $c = \rho_z'$, then we have
\begin{align}\label{eq:arb2}
\frac{d^2}{dt^2}I\circ\rho(0) = 2(a^2 + b^2 + c^2 - 2ab - 2ac - 2bc)(0).
\end{align}
From \cite{Mei2013} we know that the tangent space of $\rho$ (in a sufficiently small neighborhood of $P_1$) lies in the exterior of the null cone of the quadratic form \eqref{eq:arb2}. In particular, this gives $\frac{d^2}{dt^2}I\circ\rho(0)\neq 0$.

Let us denote by $W^{ss}$ the so-called \textit{strong-stable manifold} of the point $P_1 = (1,1,1)$ in $\mathbb{S}_0$ (the center part of the Cayley cubic), which is a smooth (in fact analytic) curve the points of which converge exponentially fast to $P_1$ under iteration of the trace map. In fact, this manifold is formed by the tangential intersection of the center-stable manifold that contains the curve of period-two periodic points with the surface $\mathbb{S}_0$. Let $F$ be a fundamental domain along $W^{ss}$ (that is, $F$ is an interval along $W^{ss}$ with endpoints $\alpha$ and $\beta$, such that $\alpha$ is mapped to $\beta$ under the trace map) in a small neighborhood of $P_1$, and let $\gamma$ be a twice differentiable curve lying in the center-stable manifold and crossing $F$ transversally. It is known that the center-stable manifold intersects transversally the surfaces $S_{V > 0}$ and does not intersect any $S_{V < 0}$ (see \cite[Proposition 4.9, 4.]{Yessen2011} for a more general result). The center-stable manifold in question, let us call it $W^{cs}$, is smoothly foliated by $\set{W^{cs}\cap S_V}_{V\geq 0}$ (see \cite{Pugh1997}), and since $\gamma$ is transversal to $F$, it follows that $\gamma$ is also transversal to $W^{cs}\cap S_V$ for all $V > 0$ sufficiently small, and hence also to the surfaces $S_{V > 0}$ near $F$. Now $\gamma$ can be parameterized similarly to $\rho$ by projecting $\gamma$ onto an arc of $\rho$ along the smooth foliation $\set{W^{cs}(\rho)\cap S_{V}}_{V > 0}$, and the claim follows from equations \eqref{eq:arb1} and \eqref{eq:arb2}.
\end{proof}

Let us now record, for ease of reference, the following result as a proposition.

\begin{prop}\label{prop:transversality}
For a given substitution sequence $v$, there exists $\delta > 0$, such that for all $(\mathfrak{p},\mathfrak{q})$ satisfying $0 < \norm{(\mathfrak{p},\mathfrak{q}) - (1,0)} < \delta$, the curve of initial conditions $l_{v, (\mathfrak{p},\mathfrak{q})}$ intersects the center-stable manifolds transversally.
\end{prop}

The proof of the above proposition can be given by following verbatim the arguments in the proof of \cite{Yessen2011a}, keeping in mind Lemma \ref{lem:qinter} from above. Indeed, the only needed ingredient is the quadratic tangency with $\mathbb{S}_0$ of the center-stable manifolds that contain the curves of periodic points passing through the singularities (notice that in \cite{Yessen2011a} and here the curve of initial conditions is the same, although in \cite{Yessen2011a} the expression for the curve was simplified by taking the inverse image of the curve twice under the Fibonacci trace map; this allowed to avoid notational and some technical complications, but is not at all necessary. One only needs to notice that the construction of the invariant cones that contain the tangent bundle of the curve of initial conditions carries over to the present case verbatim. Surely, in \cite{Yessen2011a} it was shown that the invariant cones contain the tangent vectors to the simplified version of the curve of initial conditions, but by invariance of the cones, the tangent vectors of the twice forward image of the simplified curve---that is, the original curve---must also fall into the invariant cone field.).

Now the conclusion of the theorem follows by transversality much like in the Schr\"odinger case; in particular, the arguments from \cite[Theorem 1.5]{Damanik2010a} carry over verbatim.

\subsubsection{Proof of Theorem \ref{thm:mei-yessen-4}}

Continuity and analyticity statements of the theorem can be proved by applying the arguments from the proof of \cite[Theorem 2.3, iii and iv]{Yessen2011a} without change. Thus the only statement of Theorem \ref{thm:mei-yessen-4} that requires proof here is \eqref{eq:y1}.

It is known that for $\alpha$ in the spectrum $\Sigma_{v,(\mathfrak{p},\mathfrak{q})}$ of $H_{v,(\mathfrak{p},\mathfrak{q})}$, (see \cite[Section 6]{Cantat2009}) we have
\begin{align}\label{eq:loc-dim}
\lhdim(\Sigma_{v,(\mathfrak{p},\mathfrak{q})}, \alpha) = \frac{1}{2}\hdim(\Omega_V),
\end{align}
where $\Omega_V$ is the nonwandering set for the trace map $T_v$ on the surface $V = I(l_{v, (\mathfrak{p},\mathfrak{q})}(\alpha))$ (see the proof of Theorem 2.3 in \cite{Yessen2011a}). On the other hand, from \cite{Damanik2010a} we know that there exists $V_0 > 0$ such thatfor all $V \in (0, V_0)$, there exist constants $C_1, C_2 > 0$ such that the following bound holds
\begin{align*}
1 - C_1V\leq \frac{1}{2}\hdim(\Omega_V)\leq 1 - C_2V.
\end{align*}
On the other hand, we have
\begin{align*}
\hdim{\Sigma_{v,(\mathfrak{p},\mathfrak{q})}} = \max_{\alpha\in\Sigma_{v,(\mathfrak{p},\mathfrak{q})}}\set{\lhdim(\Sigma_{v,(\mathfrak{p},\mathfrak{q})}, \alpha)}\leq 1.
\end{align*}
Computation of the value of the Fricke-Vogt invariant along $l_{v,(\mathfrak{p},\mathfrak{q}})$ is given in equation \eqref{eq:Ici} above, from which it is evident that for all $(\mathfrak{p},\mathfrak{q})$ sufficiently close to $(1,0)$, for all $\alpha\in \Sigma_{v,(\mathfrak{p},\mathfrak{q})}$, $I(l_{v,(\mathfrak{p},\mathfrak{q})}) < V_0$ (certainly $\Sigma_{v, (\mathfrak{p},\mathfrak{q})}$ is compact, and its bounds depend continuously on $(\mathfrak{p},\mathfrak{q})$, approaching $\pm 2$ as $(\mathfrak{p},\mathfrak{q})$ approaches $(1,0)$). More precisely, there exists $\delta > 0$ such that for all $(\mathfrak{p},\mathfrak{q})$ that lie within $\delta$ of $(1,0)$, we have $I(l_{v,(\mathfrak{p},\mathfrak{q})}(\alpha)) < V_0$ for all $\alpha\in\Sigma_{v,(\mathfrak{p},\mathfrak{q})}$. It then follows that for all such pairs $(\mathfrak{p},\mathfrak{q})$ and $\alpha\in\Sigma_{v, (\mathfrak{p},\mathfrak{q})}$,
\begin{align*}
1 - C_1I(l_{v,(\mathfrak{p},\mathfrak{q})}(\alpha)) \leq \lhdim(\Sigma_{v,(\mathfrak{p},\mathfrak{q})}, \alpha),
\end{align*}
leading to
\begin{align*}
1 - C_1\max_{\alpha\in\Sigma_{v, (\mathfrak{p},\mathfrak{q})}}\set{I(l_{v,(\mathfrak{p},\mathfrak{q})})(\alpha)} \leq \hdim(\Sigma_{v, (\mathfrak{p},\mathfrak{q})}).
\end{align*}
On the other hand, it is clear from equation \eqref{eq:Ici} that $I$ depends smoothly on $(\mathfrak{p},\mathfrak{q})$ in a neighborhood of $(1,0)$, so that there exists $D > 0$ such that for all pairs $(\mathfrak{p},\mathfrak{q})$ within $\delta$ of $(1,0)$, and for all $\alpha\in\Sigma_{v,(\mathfrak{p},\mathfrak{q})}$,
\begin{align*}
\max_{\alpha\in\Sigma_{v, (\mathfrak{p},\mathfrak{q})}}\set{I(l_{v,(\mathfrak{p},\mathfrak{q})})(\alpha)} \leq D\norm{(\mathfrak{p},\mathfrak{q}) - (1,0)},
\end{align*}
leading to
\begin{align*}
1 - C_1D\norm{(\mathfrak{p},\mathfrak{q}) - (1,0)}\leq 1 - C_1\max_{\alpha\in\Sigma_{v, (\mathfrak{p},\mathfrak{q})}}\set{I(l_{v,(\mathfrak{p},\mathfrak{q})})(\alpha)}.
\end{align*}
This completes the proof.

\subsubsection{Proof of Theorem \ref{thm:mei-yessen-5}}

Let us consider the gap that opens at a point $p$ on $l_{v, (1, 0)}\cap\mathbb{S}_0$ as soon as $(\mathfrak{p},\mathfrak{q})$ is turned on (i.e. shifts away from the point $(1,0)$). Let us call this gap $U_p$ and its boundary points $l_p(t)$ and $r_p(t)$ (which of course depend on the perturbation parameter $t$ -- the same $t$ that $\mathfrak{p}$ and $\mathfrak{q}$ depend on). In fact, the gaps open at countably many points along $l_{v, (1,0)}$, which are precisely the intersection points of $l_{v, (1,0)}$ with the strong-stable manifold, $W^{ss}(P_1)$ of $P_1$.

If $\mathcal{U}$ is a small neighborhood in $\R^3$ of the point $p$, then the intersection of the center-stable manifold that contains the curve of period-two periodic points passing through $P_1$, call it $W^{cs}$, with $\mathcal{U}$ gives countably many smooth connected injectively immersed two-dimensional submanifolds of $\R^3$ (each containing a point of $l_{v, (0,1)}\cap W^{ss}(P_1)\cap\mathcal{U}$, and only one for each such point), and $l_p$ and $r_p$ are the points of intersection of $l_{v, (\mathfrak{p},\mathfrak{q})}$ with the connected component that contains the point $p$; let us call this connected component $W_p$.

Since the center-stable manifolds are smooth, the perturbation $\alpha$ is smooth, and $l_{v,\alpha(t)}$ intersects the center-stable manifolds transversally for all $t$ sufficiently close to zero, it is clear that the gap boundary points, $l_p(t)$ and $r_p(t)$, depend smoothly on $t$.

Now let $\Gamma$ be a plane through the point $p$ orthogonal (at $p$) to $\mathbb{S}_0$ and intersecting $W_p$ transversally. Define a map $\Psi(x): \mathcal{U}\rightarrow\Gamma$ as follows.
\it
\begin{itemize}
\setlength{\itemsep}{2pt}
\item At each point $i$ of the curve of the intersection $\Gamma\cap\mathcal{U}\cap\mathbb{S}_0$, pick a plane $\Pi_i$ orthogonal to $\mathbb{S}_0$ at $i$. Then there exists $\delta > 0$ such that for every $V\in(-\delta,\delta)$, for each $i$, $\Pi_i\cap \mathcal{U}\cap S_V$ is a smooth curve, say $C_{i, V}$. Now for each $x\in\mathcal{U}$, there exist $i$ and $V$ such that $x\in C_{i,V}$. Let $\Psi(x)$ denote the projection of $x$ along $C_{i,V}$ onto $\Gamma$.
\end{itemize}
\rm
Observe that $\Psi$ is a smooth projection.

From now on, by abuse of notation, let us denote by $\Gamma$ the part of the plane $\Gamma$ that lies in $\mathcal{U}$, i.e. $\Gamma\cap\mathcal{U}$. For the convenience of avoiding some technical complications, but not at all out of necessity, let us apply a diffeomorphism $\Phi: \Gamma\rightarrow\Gamma$ that rectifies the curves $\Gamma\cap S_V$, $V\in (-\delta,\delta)$, and introduce the natural coordinates: the $x$ direction is along the curve $\Gamma\cap S_0$, and the $y$ direction is perpendicular to the $x$ direction and is given by the value of the Fricke-Vogt invariant, $V$ (i.e so that the curve $\Gamma\cap S_V$, after it has been rectified, crosses the $y$ axis at the point $y = V$). Position the origin on the $x$-axis at the point $p$. Denote the $x$ axis by $\mathbf{x}$ and the $y$ axis by $\mathbf{y}$. From now on when we write $\Gamma$, we assume it to come with this coordinate system (after the rectification $\Phi$ has been applied).

Notice that the curve $W^{cs}\cap \Gamma$ (the intersection is taken first and $\Phi$ is applied next) can be viewed as a graph of a function $\mathfrak{s}:\mathbf{x}\rightarrow\mathbf{y}$, and this function is quadratic near $p$ (i.e $\mathfrak{s}'(p) = 0$ and $\mathfrak{s}''(p) > 0$). In particular, it follows that near $p$, the graph of $\mathfrak{s}$ intersects the level curves $\set{y = c}$, $c$ near zero, in two points, transversally.

Denote by $\tilde{l}_t$ the segment along the curve $\Phi\circ\Psi(l_{v, \alpha(t)})$ that is bounded between the two points $\Phi\circ\Psi(l_p(t))$ and $\Phi\circ\Psi(r_p(t))$. This curve may be viewed as the graph of a function $\mathfrak{u}_t: \mathbf{x}\rightarrow\mathbf{y}$, and it is evident from \eqref{eq:fv-de} that $\abs{\mathfrak{u}'_t}\rightarrow 0$ as $t\rightarrow 0$. It follows that there exist $C_1(t), C_2(t) > 0$, $C_1(t), C_2(t)\rightarrow 1$ as $t\rightarrow 0$, such that the following holds.
\it
\begin{itemize}
\setlength{\itemsep}{2pt}
\item Denote by $L(t)$ the length of the curve $\tilde{l}_t$. Denote by $d_l(t)$ (respectively $d_r(t)$) the distance between the point $\Phi\circ\Psi(l_p(t))$ and the other point on the graph of $\mathfrak{s}$ at the same $y$-level as $l_p(t)$ (respectively, the distance between the point $\Phi\circ\Psi(r_p(t))$ and the other point on the graph of $\mathfrak{s}$ at the same $y$-level as $r_p(t)$). Then
\begin{align*}
C_1(t)d_l(t) \leq L(t) \leq C_2(t)d_r(t).
\end{align*}
(we assume, without loss of generality, that $d_l(t) \leq d_r(t)$).
\end{itemize}
\rm

It follows that in order to establish the desired limit, it is of course enough to prove that
\begin{align*}
\lim_{t\rightarrow 0}\frac{d_{l,r}(t)}{\norm{\alpha(t) - (1,0)}}\hspace{2mm}\text{ exists and belongs to }\hspace{2mm}(0,\infty).
\end{align*}
This result (rather, its slight reformulation) is available as \cite[Lemma 3.2]{Damanik2010} for the Fibonacci substitution, and is extended in \cite[Lemma 6.2]{Mei2013} for to the general case.

\subsubsection{Proof of Theorem \ref{thm:mei-yessen-7}}

Every trace map that we consider here (i.e. one corresponding to a primitive invertible two-letter substitution) can be written as a finite composition of the maps $t_a$ where for $a\in\N$,
\begin{align*}
t_a = U^a\circ P,
\end{align*}
with
\begin{align*}
U(x,y,z) = (2xz - y, x, z)\hspace{2mm}\text{ and }\hspace{2mm}P(x,y,z) = (x, z, y).
\end{align*}
In fact, for every $a\in\N$, $t_a$ is also a trace map corresponding to a substitution realized as a circle rotation by the irrational number $\alpha \in (0, 1)$ with the constant sequence $\set{a}$ forming its continued fraction expansion. If $\alpha \in (0,1)$ is an irrational number with a periodic continued fraction expansion of the form $[0;a_1\dots, a_n, a_1, \dots, a_n,\dots]$, and $s$ is the substitution that is realized as the circle rotation by $\alpha$, then $T_s$ is given by $t_{a_n}\circ t_{a_{n-1}}\circ\cdots\circ t_{a_1}$ (see \cite{Mei2013}).

In general, as remarked in the introduction, the substitutions that we deal with here can be realized as circle rotations by an irrational angle with \textit{eventually periodic} continued fraction expansion. However, we can follow \cite{Mei2013} to reduce the problem to a \textit{periodic} case. Indeed, suppose that for an $\alpha\in(0, 1)$, its continued fraction expansion is given by $[0;b_1, \dots, b_k, a]$, where $a$ is a periodic sequence with period $a_1, \dots, a_n$. Let us denote the corresponding substitution on two letters by $s$. Denote by $T_b$ the map $t_{b_k}\circ\cdots\circ t_{b_1}$ and by $T_a$ the map $t_{a_n}\circ\cdots\circ t_{a_1}$. For a letter $\star$, assuming, as in the introduction, that $s(\star)$ begins with $\star$, for any $m\in\N$, the trace of the transfer matrices over $s^m(\star)$ sites can be computed by applying $T_a^m\circ T_b$ to the initial conditions (the initial three traces, as above). Now if $l$ denotes the curve of initial conditions, as above, then we can consider the new curve $T_b(l)$ and run the trace map $T_a$ on this curve. For further detail, see \cite{Mei2013}. Note also that the discussion in this paragraph is more general than is needed for our situation, in the sense that for a primitive and invertible substitution, the length of the initial sequence $b_1,\dots,b_k$ is always at most $1$ (see \cite{Mei2013} and references therein).

A direct computation, which we omit here (but see \cite{Mei2013}) confirms that for every $a\in\N$, $t_a(P_1) = P_1$ and the spectrum of $Dt_a(P_1)$ is $\set{\lambda, 1, \mu}$ with $\abs{\lambda} < 1 <\abs{\mu}$. One of the ingredients that we need for this proof is that the eigenvector corresponding to $\lambda$ be transversal to the plane $\set{z = 1}$. After showing this, we shall show how the proof of the analogous theorem in the case of the Fibonacci substitution from \cite{Yessen2011a} can be used essentially without modification.

Let us define the matrices
\begin{align*}
M_a \eqdef \begin{pmatrix}
	a	&	1\\
	1	&	0\\
	\end{pmatrix}
\hspace{2mm}
\text{ with }a\in\N.
\end{align*}
Clearly $\det(M_a) = -1$, while its eigenvalues are real and not equal to $\pm 1$. In other words, $M_a$ defines an Anosov automorphism on the two-torus $\mathbb{T}^2$. It is also not difficult to see that for any $n > 1$, and $a_1, \dots, a_n\in\N$, the product
\begin{align*}
M_{a_1}M_{a_2}\cdots M_{a_n}
\end{align*}
also defines an Anosov automorphism on $\mathbb{T}^2$ (indeed, it is easy to see that the trace of the product is strictly larger than $1$, while the determinant is $\pm 1$ and each matrix in the product is symmetric, forcing the eigenvalues to be real and not equal to $\pm 1$, which in turn implies that one eigenvalue is strictly larger than $1$ in absolute value, while the other is strictly smaller than $1$ in absolute value). Now by generalizing the argument from (Lemma 3.15 in \cite{Mei2013}) (a computation which we omit here), it isn't difficult to see
\begin{align*}
F\circ M_{a_1}M_{a_2}\cdots M_{a_n} = T_{a_1}\circ T_{a_2}\circ\cdots\circ T_{a_n}|_{\mathbb{S}_0},
\end{align*}
with $F$ being the factor map from \eqref{eq:factor}. Let us denote the product of $M_{a_i}$ as above by $M$ and the composition of $T_{a_i}$ by $T$. If we let $U$ denote an arbitrarily small neighborhood of the singularities $P_1,\dots, P_4$, then we see that the stable vector field on $\mathbb{T}^2\setminus F^{-1}(U)$ for $M$ is mapped onto the stable vector field for $T$ on $\mathbb{S}_0\setminus U$ by the differential of $F$, $DF$ (by the \textit{stable vector field} we mean the vector field of the (normalized) eigenvectors corresponding to the smaller eigenvalue; i.e. the contracting vector field). Indeed, the differential vanishes at the preimages of the singularities, but it is still ``sufficiently conformal'' to preserve the uniform transversality of the stable and the unstable vector fields (see \cite[Lemma 3.1]{Damanik2009} for the formal statement).

Now notice that $F(0,0) = P_1$, and in the neighborhood of $(0,0)\in\mathbb{T}^2$, $DF$ is given by
\begin{align*}
DF_{(\theta,\phi)} = \begin{pmatrix}
	-2\pi\sin2\pi(\theta + \phi)	&	-2\pi\sin 2\pi(\theta + \phi)\\
	-2\pi\sin2\pi\theta				&	0\\
	0								&	-2\pi\sin2\pi\phi\\
\end{pmatrix}.
\end{align*}
It isn't difficult to see that the matrix $M$ does not have any eigenvectors $(u, v)$ with $v = 0$. Now let $(u,v)$ be an eigenvector of $M$ at the point $(0, 0)$ in the stable direction. Let us also normalize the vector so that $v = 1$. For $\epsilon > 0$ small, $(\epsilon u, \epsilon)$ is a point on the stable manifold of the point $(0,0)$ in a small (of order $\epsilon$) neighborhood of $(0,0)$. This point is mapped by $F$ onto a point on the strong-stable manifold of $P_1$ in $\mathbb{S}_0$, and the vector $(u, 1)$ is mapped by $DF_{(\epsilon u, \epsilon)}$ to a vector tangent to the strong-stable manifold of $P_1$ at the point $F(\epsilon u, \epsilon)$ (which lies in a small, of order $\epsilon$, neighborhood of $P_1$). We have
\begin{align*}
\begin{pmatrix}
	v_x\\
	v_y\\
	v_z\\
\end{pmatrix}\eqdef DF_{(\epsilon u, \epsilon)}(u, 1) =
\begin{pmatrix}
	-(2\pi u + 2\pi)\sin2\pi(\epsilon u + \epsilon)\\
	-2\pi u\sin2\pi\epsilon u\\
	-2\pi \sin2\pi\epsilon\\
\end{pmatrix}.
\end{align*}
Now we see that
\begin{align}\label{eq:lim}
\lim_{\epsilon\rightarrow 0}\frac{v_x^2 + v_y^2}{v_z^2} \in (0, \infty).
\end{align}
This shows that in a neighborhood of $P_1$, the line that is spanned by the vector tangent to the strong-stable manifold of $P_1$ on $\mathbb{S}_0$ is uniformly (in $\epsilon$, for all $\epsilon > 0$ sufficiently small) transversal to the plane $\set{z = 1}$. Thus transversality of the strong-stable manifold of $P_1$, $W^{ss}(P_1)$, with the plane $\set{z=1}$ also holds at the point $P_1$. Now one could apply the arguments from \cite[Theorem 2.3]{Yessen2011a} with slight modifications, as the geometric setup is now essentially the same, except for the curve of initial conditions. However, we can take the current setup closer to that of \cite{Yessen2011a} as follows.

Let $f(x,y,z) \eqdef (2xy - z, x, y)$, the Fibonacci trace map. Let us change the coordinates in $\R^3$ by applying $f^{-2}$ (the two-fold composition of the inverse of $f$). In these new coordinates, the curve of initial conditions $l_{v,(\mathfrak{p},\mathfrak{q})}$ maps to the line
\begin{align*}
\left(\frac{E - \mathfrak{q}}{2}, \frac{E}{2\mathfrak{p}}, \frac{1 + \mathfrak{p}^2}{2\mathfrak{p}}\right).
\end{align*}
Furthermore, in the new coordinates the transversality at $f^{-2}(P_1) = P_1$ of the $f^{-2}(W^{ss}(P_1))$ with the plane $\set{z = 1}$ holds, which can be easily checked by computing $Df^{-2}$ and then applying it to the vector $(v_x, v_y, v_z)^T$ and noting that the resulting vector has the $z$ component bounded away from zero. Now the proof is finished by applying the proof of \cite[Theorem 2.3]{Yessen2011a} without modification.

\subsubsection{Proof of Theorem \ref{thm:mei-yessen-8}}

The proof here can be given by repeating the proof of \cite[Theorem 2.5]{Yessen2011a} without modification. Indeed, one only needs to know that the intersection of the curve of initial conditions with the center-stable manifolds is transversal for all $(\mathfrak{p},\mathfrak{q})$ in the regions described in the statement of Theorem \ref{thm:mei-yessen-8}, which can be insured just as in \cite{Yessen2011a}.

\subsubsection{Proof of Theorem \ref{thm:mei-yessen-9}}

The proof of \cite[Theorem 2.6]{Yessen2011a} applies here without modification (after a coordinates change via $f^{-2}$ just like in the proof of Theorem \ref{thm:mei-yessen-7}).

\subsubsection{Proof of Theorem \ref{thm:p-zero}}

Let us fix a substitution sequence $v$ and drop it from notation, as above. Observe that for any fixed $\mathfrak{q}$ and any $\mathfrak{p}$ with $\abs{\mathfrak{p}}\in(0,1)$, for any $V > 0$, $I\circ l_{(\mathfrak{p},\mathfrak{q})}(E) \leq V$ implies (by equation \eqref{eq:Ici})
\begin{align*}
E\geq \frac{4V\mathfrak{p}^2 - \mathfrak{q}^2 - (\mathfrak{p}^2-1)^2}{\mathfrak{p^2}-1}\underset{\mathfrak{p}\rightarrow 0}{\longrightarrow} \mathfrak{q}^2 + 1,
\end{align*}
and if $\mathfrak{p}\ll \min\set{1, \mathfrak{q}^2}$, such $E$ cannot belong to the spectrum, since the spectrum lies in $\left[-\norm{H_{(\mathfrak{p},\mathfrak{q})}}, \norm{H_{(\mathfrak{p},\mathfrak{q})}}\right]$, and $\norm{H_{(\mathfrak{p},\mathfrak{q})}} < \mathfrak{q}^2+1$. Hence there exists $V: (-1, 1)\rightarrow (0, \infty)$ satisfying $V(\mathfrak{p})\rightarrow \infty$ as $\mathfrak{p}\rightarrow 0$, such that for all $\mathfrak{p}$ sufficiently close to $0$, $I\circ l_{(\mathfrak{p}, \mathfrak{q})}(\Sigma_{(\mathfrak{p},\mathfrak{q})})\subset [V(\mathfrak{p}), \infty)$. Now as a consequence of Theorem \ref{thm:degt}, we obtain the desired result.

\section*{Acknowledgement}

It is a pleasure for us to thank Anton Gorodetski and David Damanik for many helpful discussions and suggestions. We would also like to thank Yuki Takahashi for his questions that led to Theorems \ref{thm:p-zero} and \ref{thm:p-zero2}. W.Y. would also like to thank Jean Bellissard for helpful discussions.

\bibliographystyle{amsplain}
\bibliography{bibliography}

\end{document}